\documentclass[11pt,a4paper]{article}
\usepackage{graphicx,nicefrac}
\usepackage{tabularx,booktabs,multirow}
\usepackage{todonotes}
\usepackage{amssymb,amsmath,amsthm}
\usepackage{mathtools}
\usepackage[pagebackref]{hyperref}
\usepackage{subcaption}
\usepackage{amsmath,amsfonts}
\usepackage{xcolor}
\usepackage{xspace}
\usepackage{bbm}
\usepackage{tikz}
\usepackage{makecell}
\usepackage{authblk}
\usepackage{fullpage}
\usepackage{graphicx}
\usepackage{cleveref}
\usepackage{bm}
\usepackage{bigfoot}
\usepackage[T1]{fontenc}
\usepackage{graphicx}
\usepackage{comment}
\usepackage{nicefrac}

\usepackage{amsmath}
\usepackage{amssymb}
\usepackage{tabularx}
\usepackage[]{todonotes} 
\usepackage{multirow}

\usepackage{tikz}
\usetikzlibrary{arrows}
\usetikzlibrary{automata,arrows,positioning,calc}

\newcommand{\T}{T}
\usepackage{thmtools} 
\usepackage{thm-restate}

\usetikzlibrary{calc,decorations.markings}
\tikzstyle{vertex}=[draw, circle, fill, inner sep = 2pt]
\tikzstyle{squared-vertex}=[draw, fill, inner sep = 2pt]
\tikzstyle{approves}=[decoration={
	markings,
	mark=at position 0.8 with {\arrow{>}}}]

\makeatletter
\def\moverlay{\mathpalette\mov@rlay}
\def\mov@rlay#1#2{\leavevmode\vtop{%
		\baselineskip\z@skip \lineskiplimit-\maxdimen
		\ialign{\hfil$\m@th#1##$\hfil\cr#2\crcr}}}
\newcommand{\charfusion}[3][\mathord]{
	#1{\ifx#1\mathop\vphantom{#2}\fi
		\mathpalette\mov@rlay{#2\cr#3}
	}
	\ifx#1\mathop\expandafter\displaylimits\fi}
\makeatother

\newcommand{\cupdot}{\charfusion[\mathbin]{\cup}{\cdot}}

\newcommand{\Gla}{\ensuremath{G_{\ell - \alpha + 1}}}

\newtheorem{corollary}{Corollary}
\newtheorem{proposition}{Proposition}

\newtheorem{example}{Example}

\crefname{observation}{Observation}{Observation}

\title{Stable Matching with Multilayer Approval Preferences: Approvals can be Harder than Strict Preferences}

\author{Matthias Bentert} 
\author{Niclas Boehmer}
\author{Klaus Heeger}
\author{Tomohiro Koana}

\affil{\small
  Technische Universit\" at Berlin, Algorithmics and Computational Complexity\protect\\
  \{matthias.bentert,niclas.boehmer,heeger,tomohiro.koana\}@tu-berlin.de}

\date{\today}

\begin{document}
\maketitle              
\begin{abstract}
We study stable matching problems where agents have multilayer preferences:
There are~$\ell$~layers each consisting of one preference relation for each agent.
Recently, Chen et al.~[EC '18] studied such problems with strict preferences, establishing four multilayer adaptions of classical notions of stability.
We follow up on their work by analyzing the computational complexity of stable matching problems with multilayer \emph{approval} preferences.
We consider eleven stability notions derived from three well-established stability notions for stable matchings with ties and the four adaptions proposed by Chen et al.
For each stability notion, we show that the problem of finding a stable matching is either polynomial-time solvable or NP-hard.
Furthermore, we examine the influence of the number of layers and the desired ``degree of stability'' on the problems' complexity. 
Somewhat surprisingly, we discover that assuming approval preferences instead of strict preferences does not considerably simplify the situation (and sometimes even makes polynomial-time solvable problems NP-hard).
\end{abstract}

\section{Introduction}
Problems related to matching under preference are a popular and extensively researched topic in computer science, economics, and mathematics \cite{DBLP:books/ws/Manlove13}. 
In the classical \textsc{Stable Marriage} problem, we are given two sets of agents with each agent having strict preferences over the agents from the other side. A matching of agents from one side to the other is \emph{(Gale-Shapley) stable} if there is no so-called blocking pair, i.e., a pair of agents preferring each other to their current partner.
However, in reality, agents may rank the other agents with respect to multiple criteria, with each of these criteria giving rise to a different evaluation of agents.
Motivated by this, Chen et al. \cite{DBLP:conf/sigecom/ChenNS18} pioneered the study of \textsc{Stable Marriage} where agents have multilayer preferences. 
In their model, there are~$\ell$~separate layers, and in each layer, all agents provide a strict ranking of agents from the other side. 
Thus, each agent specifies $\ell$ strict rankings, one for each layer.

Multilayer preferences are a general framework which can model a wide range of situations: 
for instance, a layer may represent a criterion according to which agents evaluate each other. 
Another example concerns uncertain situations:
Here, scenarios for the future each give rise to a separate layer containing the agents' preferences in this scenario. 
Lastly, matching fixed groups to each other (e.\,g.\ couples or classes), preferences of the different group members may be expressed in multiple layers: Each agent represents one group and the preferences in one layer represent the preferences of one (arbitrary) group~member.

Chen et al.\@ \cite{DBLP:conf/sigecom/ChenNS18} considered four different multilayer adaptions of Gale-Shapley stability and showed for each that deciding the existence of a stable matching is NP-hard even for only four layers.
Motivated by this, we study a simpler preference model: multilayer \emph{approval} preferences. Here, in each layer, instead of providing a strict ranking, agents approve some agents and disapprove all others.\footnote{In our model, agents \emph{prefer} agents they approve to agents they disapprove and to having no partner, but are \emph{indifferent} between the later two.}
Moving from strict to approval preferences gives us more options for stability notions.
We study adaptions of the three established stability notions for stable matchings with ties and single-layered preferences: \emph{weak}, \emph{strong}, and \emph{super stability} \cite{DBLP:journals/dam/Irving94,DBLP:books/ws/Manlove13}.
For instance, under the most popular \emph{weak stability} criterion, a matching is stable if there is no pair of agents that both prefer each other to their current partner. 
As approval preferences are easier to cast, are often used in practice, and typically lead to better axiomatic guarantees and faster algorithms, they have already been widely applied to collective-decision problems (see for example~\cite{DBLP:journals/teco/AzizBM20,DBLP:journals/jair/BouveretL08,DBLP:journals/corr/abs-2007-01795,DBLP:conf/aaai/TalmonF19}; for voting, approval preferences even constitute their own subfield). 
However, approval preferences have only rarely been considered in the context of matchings under preferences because stable matchings are trivial to construct for single-layer approvals.\footnote{E.\,g., a weakly stable matching corresponds to a maximal matching in the undirected part of the approval graph (see \Cref{se:prelims}). }
However, (multilayer) approval preferences are of no less practical relevance in matching markets, as they arise, for instance, if preferences
model compatibility, availability, or simply whether agents
have a certain attribute or qualification (from the other agent's
perspective). For instance, when matching students for a group homework, layers could represent weekdays and two agents approve each other in some layer if they are both available on this day of the week. Alternatively, layers could represent whether students deem each other qualified with respect to different criteria (e.\,g., being able to solve the homework, writing down the homework, presenting the homework) or whether they have completed certain relevant previous courses. 

To adapt weak, strong, and super stability to multilayer preferences, we use the following four generalizations proposed by Chen et al. \cite{DBLP:conf/sigecom/ChenNS18} (the meaning of ``favors'' and ``stability'' in the following definitions depends on whether we generalize weak, strong, or super stability; see \Cref{se:prelims} for details).
A matching is \emph{all-layers stable} if it is stable in every layer.
The other three generalizations are each equipped with a desired degree $\alpha$ of stability. 
A matching is \emph{$\alpha$-globally stable} if there are $\alpha$ layers in which the matching is stable.
(In particular, all-layers stability is equivalent to $\ell$-global stability.)
A matching is \emph{$\alpha$-pair stable} if for each unmatched pair of agents, there are $\alpha$ layers in which one of the two agents ``favors'' the current matching to the other agent, i.e., each pair may block in at most $\ell-\alpha$ layers.  
Lastly, in an \emph{$\alpha$-individually stable} matching, for each unmatched pair of agents, one of them ``favors'' the current matching to the other agent in at least $\alpha$ layers.  
Which of these four generalizations is appropriate depends not only on the application but also on what the different layers represent: 
For instance, if each layer captures a potential state of an uncertain scenario, then $\alpha$-global stability may be useful to maximize the probability that stability is established. 
In contrast, when each agent models a group and each layer contains the preferences of a group member, $\alpha$-global stability seems less appealing, as a priori the preferences of different agents within one layer are completely unrelated. 
Here, $\alpha$-individual stability is more natural, as in an $\alpha$-individually stable matching, asking two currently unmatched groups whether they prefer being together, in one of them at least $\alpha$ agents vote against this. 

Combining weak, strong, and super stability with the four multilayer generalizations of Chen et al. \cite{DBLP:conf/sigecom/ChenNS18}, we analyze the computational complexity of deciding the existence of a stable matching for eleven stability notions,\footnote{For individual stability, we have only two stability notions depending on the definition of when one agent ``favors'' another.} 
also taking into account two natural preference restrictions: symmetric approvals, where agents' approvals are mutual, and bipartite approvals, where there is a bipartition of the agents and each agent only approves or disapproves agents from the other set. Symmetric approvals arise for instance if preferences encode compatibility constraints, while in many matching markets approvals are by-design bipartite, e.\,g., when matching applicants to jobs, students to schools, or mentees to mentors.

\subsection{Related Work}
Chen et al. \cite{DBLP:conf/sigecom/ChenNS18} proved that deciding whether an $\alpha$-globally/pair stable matching exists is NP-hard for any $2\leq \alpha \leq \ell$ for \textsc{Stable Marriage} with strict preferences.
For individual stability they proved that the problem is polynomial-time solvable for $\alpha=\ell$ but NP-hard for~$2 \le \alpha \le \frac{2}{3} \ell$. 
Moreover, they identified two preference restrictions that lead to polynomial-time solvability: for $\alpha$-global stability if, within each layer, all agents from one side have the same preferences, and for $\alpha$-pair and $\alpha$-individual stability with $\alpha>\lfloor \ell / 2 \rfloor$ if the preferences of agents from one side do not change between different layers.
In sum, our work differs from Chen et al.'s work in the following points:
We consider approval instead of strict preferences (leading to eleven algorithmic questions, which are all fundamentally different from the work of Chen et al.), we do not restrict ourselves to the bipartite case, and we study several new parameterizations to achieve tractability.
Notably, problems equivalent to deciding whether a \textsc{Stable Marriage} instance with multilayer strict preferences admits an all-layers stable matching have also been studied by Miyazaki and Okamoto \cite{DBLP:journals/jco/MiyazakiO19} and Aziz et al.~\cite{DBLP:journals/algorithmica/AzizBGHMR20}.
Following up on the work of Chen et al. \cite{DBLP:conf/sigecom/ChenNS18}, Wen et al.~\cite{posbasedmm} also studied a bipartite matching problem where agents have multilayer strict preferences over the agents from the other side. 
However, different from our work and the work of Chen et al. \cite{DBLP:conf/sigecom/ChenNS18}, they did not consider any type of stability notion and instead studied the problem of finding a matching that minimizes different types of  ``dissatisfaction scores''.
Moreover, Steindl and Zehavi~\cite{DBLP:conf/eumas/SteindlZ21,DBLP:conf/eumas/SteindlZ21a} studied a multilayer version of the house allocation problem, where agents have multilayer preferences over a set of houses and a matching of houses to agents needs to be found. 
Notably, they considered an extension of Pareto optimality analogous to global stability.  

Multilayer preferences have started to gain increasing popularity in the area of computational social choice, e.\,g., in multiwinner voting \cite{DBLP:conf/ecai/0001T20} and fair division~\cite{DBLP:journals/mss/Suksompong18,DBLP:conf/ijcai/KyropoulouSV19,DBLP:journals/ai/Segal-HaleviS19}; and in their blue sky\ paper Boehmer and Niedermeier \cite{DBLP:conf/atal/BoehmerN21} called for a broader application of multilayer preferences in the area.
Conceptually closely related, uncertain preferences, where different preferences have different probabilities, have also been studied in the areas of stable matching \cite{DBLP:conf/atal/AzizBFGHMR17,DBLP:journals/algorithmica/AzizBGHMR20} and resource allocation \cite{DBLP:journals/ai/AzizBHR19}.
Notably, we are interested in finding a single matching for multiple preference profiles.
An ``opposite'' problem of finding a set of (proportional) matchings given a single preference profile has been studied by Boehmer et al.~\cite{DBLP:journals/corr/abs-2102-07441}, who focused on agents having symmetric/bipartite approval preferences. 

From a technical perspective, our problems can be phrased as finding a matching fulling certain properties in some multi-layer graph, i.e., a graph with multiple edge sets defined over the same vertex set. 
While there exist numerous works on multilayer graphs \cite{DBLP:conf/asunam/MagnaniR11,boccaletti2014structure,DBLP:journals/compnet/KivelaABGMP14}, only few studied matching-related problems and all of them are different from the ones considered here (see for example~\cite{bredereck_komusiewicz_kratsch_molter_niedermeier_sorge_2019}).

\begin{table*}[t]
	\caption{Overview of our results for different stability notions. 
		All algorithmic results are for arbitrary (asymmetric) approvals (except the results marked with $\dagger$), while all hardness results (except the ones marked with $\ddagger$) hold for symmetric approvals. 
		Most hardness results also hold if approvals are bipartite and each agent only approves few agents.}
	\resizebox{\textwidth}{!}{\begin{tabular}{ c|c|c|c|c } 
			
			&  all-layers & global & pair & individual \\ \hline \hline
			\multirow{2}{2em}{weak}  & \multirow{2}{10em}{NP-h. for any $\ell \ge 2$ (T. \ref{th:all-layer-weak})}  &  \multirow{2}{14em}{NP-h. for any $\ell\geq \alpha \geq 2$ (Pr. \ref{pr:global-weak})} & \multicolumn{2}{|c}{P for $\alpha\le \lceil \ell / 2 \rceil$ (T. \ref{th:weak-individual})}  \\ 
			&   &  &  \multicolumn{2}{|c}{NP-h. for any $\ell \geq 2,\alpha>\lceil \ell / 2 \rceil$  (T. \ref{weak:pair-NP})}\\ \hline
			\multirow{3}{2.8em}{strong} &
			P for sy. (Pr. \ref{th:strong-alllayer})$^{\dagger}$ & NP-h. and W[1]-h. wrt. $\alpha$  (Pr. \ref{th:strong-global}) & NP-h. for any~$\ell \geq 2$ and & \multirow{2}{2em}{$-$}\\ 
			&NP-h.  for any $\ell\geq 3$ (T. \ref{thm:all-layer-strong})$^\ddagger$ & FPT wrt. $\ell$ for sy. (Co. \ref{co:strong-globalPara})$^{\dagger}$ & any~$0 < \alpha < \ell$  (T. \ref{th:pair-strong})  & \\
			&& XP wrt. $\alpha$ for sy. (Co. \ref{co:strong-globalPara})$^{\dagger}$  &&
			\\ \hline
			\multirow{3}{2em}{super} & \multicolumn{2}{c|}{} &\multicolumn{2}{|c}{NP-h. for any $\ell \geq 2,\alpha \leq \ell / 2$ (T. \ref{th:super-pair-neg})}\\
			&\multicolumn{2}{c|}{P  (Pr. \ref{th:super-global})} & P for $\alpha\geq 2\ell / 3$ for sy. (T. \ref{th:super-pair-pos})$^{\dagger}$ & P for $\alpha > \ell / 2$ for sy. (T. \ref{th:super-pair-pos})$^{\dagger}$\\
			&\multicolumn{2}{c|}{} & FPT wrt. $\ell$ if~$\alpha > \ell / 2$ for sy. (T. \ref{th:super-pair-pos})$^{\dagger}$ &\\
	\end{tabular}}\label{tab:results} 
\end{table*}

\subsection{Our Contributions}
We conduct an extensive study of stable matching problems with multilayer approval preferences by considering four different multilayer adaptions of three traditional stability concepts.
For each resulting stability notion, we show whether deciding the existence of a stable matching is polynomial-time solvable or NP-hard; often also pinpointing the complexity for all $\alpha\leq \ell\in \mathbb{N}$. See \Cref{tab:results} for an overview.
Lastly, in \Cref{se:similarity}, we analyze two parameters measuring ``similarity'' in the agents' preferences and show that all our problems are fixed-parameter tractable with respect to both parameters.
We present three important takeaways from our results already here: 

First, while constructing stable matchings for weak, strong, and super stability is simple in the one-layer setting, we show that this task is NP-hard for nine of our eleven stability notions in the general multilayer case.
Our hardness results are quite strong, as we often show hardness in restrictive settings, e.\,g., for two-layered symmetric and bipartite approvals with each agent only approving few agents. 
As we have only two examples of questions which are polynomial-time solvable for symmetric preferences but become NP-hard for asymmetric approvals and no such example for bipartite approvals, our results suggest that these two seemingly strong restrictions do not influence the problems' complexity much.
Nevertheless, we identify some tractable cases, e.\,g., when we have different forms of ``similarity'' in the agents' preferences (\Cref{se:similarity}).

Second, from an algorithmic perspective, multilayer approval preferences are not simpler than multilayer strict preferences, rather they sometimes make problems computationally harder: 
Comparing the picture for Gale-Shapley stability for strict preferences and its most natural analogue weak stability for approvals, we identify two questions that are polynomial-time solvable for the former but NP-hard for the latter.\footnote{These are finding an $\ell$-individually stable matching (\Cref{th:pair-weakP}) and finding an $\ell$-pair stable matching for bipartite approvals where the preferences of agents on one side do not change (\Cref{th:all-layer-weak}).}  
Moreover, for strict preferences, a stable matching is guaranteed to exist for one layer, and thus for $\alpha=1$ all problems of Chen et al. \cite{DBLP:conf/sigecom/ChenNS18}  are polynomial-time solvable. 
In contrast, for approval preferences, there are (single-layer) instances with approval preferences in which there is no matching which is strongly stable or super stable.
In fact, we identify several cases which are already NP-hard for $\alpha=1$ (see \Cref{tab:results}).
However, we also find examples where approvals are ``easier'' than strict preferences, e.\,g., finding an~$\lfloor \ell / 2 \rfloor$-individually stable matching is polynomial-time solvable for weak stability and approval preferences but NP-hard for Gale-Shapely stability and strict preferences \cite{DBLP:conf/sigecom/ChenNS18}.

Third, while our complexity picture for pair and individual stability is quite similar, it is significantly different for global stability, a contrast that is seemingly not present for strict preferences. 
Moreover, weak, strong, and super stability also lead to different results, with super stability being the in some sense computationally easiest of the three (which is in line with other works on stable matchings with indifferences).

\section{Preliminaries} \label{se:prelims}

For~$i \in \mathbb{N}$, we use~$[i] = \{1,2,\ldots,i\}$. For a set $S$, we use $\binom{S}{2}$ to denote the set of all~$2$-element subsets of $S$.
\paragraph{Preferences and Matchings.}
Let $A=\{a_1,\dots, a_n\}$ be the set of agents and $\ell\in \mathbb{N}$ be the number of layers, i.e., each agent has~$\ell$~layers of preferences. 
For~$i\in [\ell]$, each agent $a\in A$ approves a subset of agents~$\T_a^i\subseteq A$ in layer~$i$. 
We say that agent $a$ \emph{approves} agent $a'$ in layer $i$ if~${a'\in \T_a^i}$.

A matching $M\subseteq \binom{A}{2}$ is a set of agent pairs where each agent appears in at most one pair.
For a matching~$M$ and an agent~$a\in A$, we say that~$a$ is \emph{matched} in~$M$ if there is an agent~$a'$ such that~$\{a,a'\}\in M$; otherwise $a$ is \emph{unmatched}. 
Further, if $a$ is matched, we denote as $M(a)$ the partner of $a$ in $M$, i.e., if~$\{a,a'\}\in M$, then~$M(a)=a'$.
If~$a$ is unmatched, then we set~$M(a):=\square$.
An agent $a$ is \emph{happy} in matching $M$ in layer~$i$ if~$M(a)\in \T_a^i$ and \emph{unhappy} otherwise.
In layer~$i\in [\ell]$, an agent $a\in A$ \emph{prefers} being matched to an agent from~$\T_a^i$ to being unmatched or matched to an agent from~$A\setminus \T_a^i$.
Moreover, $a$ is \emph{indifferent} between being matched to any agent in~$\T_a^i$, \emph{indifferent} between being matched to any agent in~$A\setminus \T_a^i$, and \emph{indifferent} between being matched to an agent in~$A\setminus \T_a^i$ and being unmatched.
For the sake of brevity, we also say that~$a\in A$ prefers $b$ to~$c$ in layer~$i$ if $b\in\T_a^i$ and~$c\in A\setminus \T_a^i$, and that $a$ is indifferent between $b$ and~$c$ in layer~$i$ if either $b\in\T_a^i$ and~$c\in\T_a^i$ or $b\in A \setminus \T_a^i$ and~$c\in A \setminus \T_a^i$.
Moreover, we say that~$a\in A$ is indifferent in layer $i$ between $b\in A \setminus \T_a^i$ and $\square$ (which represents being unmatched).
The agents' preferences in some layer~$i\in [\ell]$ can also be represented as a directed (approval) graph $G_i=(A,E_i)$ whose vertices are the agents and which contains an arc from an agent~$a$ to an agent~$a'$ if $a$ approves $a'$ in layer~$i$, i.e., $E_i=\{(a,a')\in A\times A \mid a'\in  \T_a^i \}$.
Approvals are \emph{symmetric} if an agent~$a$ approves an agent $a'$ in some layer $i$ if and only if $a'$ approves $a$ in layer~$i$.
For symmetric approvals, the approval graph for layer~${i\in [\ell]}$ can be modeled as an undirected graph.
Approvals are \emph{bipartite} if the graph~$G=(A,\cup_{i\in [\ell]} E_i)$ is~bipartite.

\paragraph{Notions of Stability and Their Relationships.}

We consider generalizations of three different established notions of stability for the single-layer setting:
Under \emph{weak stability}, an agent pair~$\{a,a'\}\in \binom{A}{2}$ blocks a matching~$M$ in layer~$i$ if both~$a$ and~$a'$ prefer each other to~$M(a)$ and~$M(a')$, respectively.
Under \emph{strong stability}, an agent pair~$\{a,a'\}\in \binom{A}{2}$ blocks a matching~$M$ in layer~$i$ if (i) $a$ prefers~$a'$ to~$M(a)$ and (ii) $a'$ prefers~$a$ to~$M(a')$ or is indifferent between $a$ and~$M(a')$ (the roles of~$a$ and~$a'$ are interchangeable).
Under \emph{super stability}, an agent pair~$\{a,a'\}\in \binom{A}{2}$ blocks a matching~$M$ in layer~$i$ if (i) $a$ prefers $a'$ to~$M(a)$ or is indifferent between~$a'$ and~$M(a)$ and (ii) $a'$ prefers~$a$ to~$M(a')$ or is indifferent between~$a$ and~$M(a')$.
A matching without a blocking pair is called stable under the respective stability notion. 
We do not specify under which stability notion a pair blocks a matching in some layer if it is clear from context. 
Note that every super stable matching is strongly stable and every strongly stable matching is weakly stable.

Weakly and strongly stable matchings can be nicely characterized as matchings in the undirected approval graph~${G_i=(A,{E}_i)}$:
A matching $M$ is weakly stable in layer $i$ if and only if $M$ restricted to ${E}_i$ is a maximal matching in~${G}_i$. 
A matching~$M$ is strongly stable in layer $i$ if and only if for each agent~$a$ that has a neighbor in ${G}_i$, it holds that~${\{a, M(a) \} \in E_i}$.

Extending the work of Chen et al. \cite{DBLP:conf/sigecom/ChenNS18}, we study four generalizations of the above described classic stability notions to the multilayer setting (see \Cref{fig:relation} for an overview of the relationship of the different notions). 
\begin{figure}[t]
	\centering
	\begin{tikzpicture}
	\tikzset{edge/.style = {->,> = latex'}}
	\node[text width=1.5cm] (A) at (1.4,0) 
	{all-layers};
	\node[text width=2.2cm] (B) at (6,0) 
	{$\ell$-individual};
	\node[text width=1.5cm] (C) at (0,-1.5) 
	{$\alpha$-global};
	\node[text width=1.2cm] (D) at (3,-1.5) 
	{$\alpha$-pair};
	\node[text width=2.2cm] (E) at (6,-1.5) 
	{$\alpha$-individual};
	\node[text width=1.4cm] (F) at (0,-3) 
	{$1$-global};
	\node[text width=1cm] (G) at (3,-3) 
	{$1$-pair};
	\node[text width=2cm] (H) at (6,-3) 
	{$1$-individual};
	\draw[edge] (B) to (A);
	\draw[edge] (A) to (C);
	\draw[edge] (A) to (D);
	\draw[edge] (C) to (F);
	\draw[edge] (D) to (G);
	\draw[edge] (B) to (E);
	\draw[edge] (E) to (H);
	\draw[edge] (F) to (G);
	\draw[edge] (C) to (D);
	\draw[edge] (E) to (D);
	
	\draw
	(G.-5) edge[auto=right] node {} (H.183)  
	(H.177) edge[auto=right] node {} (G.5);  
	\end{tikzpicture}
	\caption{Overview of the relations between different stability notions for some~${\alpha\in [\ell]}$ (see~\protect\cite{DBLP:conf/sigecom/ChenNS18}). An arc from one notion to another implies that the first implies the second. These relationships apply to the respective adaptions of weak, strong, and super stability (there is no strong individual stability tough).}\label{fig:relation} 
\end{figure}
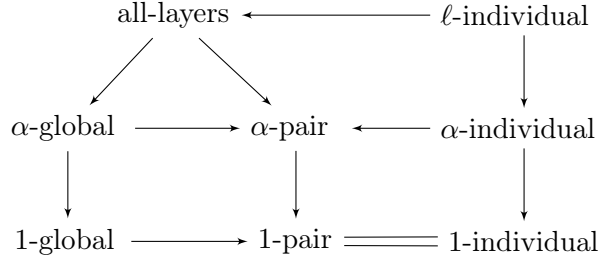%
A matching~$M$ is~\emph{$\alpha$-globally weakly/strongly/super stable} if there is a subset~$S\subseteq [\ell]$ of $\alpha$ layers such that $M$ is weakly/strongly/super stable in each layer from~$S$. 
A matching is \emph{all-layers stable} if it is~$\ell$-globally stable.
Moreover, a matching $M$ is \emph{$\alpha$-pair weakly/strongly/super stable} if for each pair~$\{a,a'\}$ of agents there is a subset~${S\subseteq [\ell]}$ of at least $\alpha$ layers where $\{a,a'\}$ is not blocking under weak/strong/super stability.
Note that for each~$\alpha\in [\ell]$, an~$\alpha$-globally stable matching~$M$ is also $\alpha$-pair stable, as in the $\alpha$ layers in which~$M$ is stable, each pair of agents is not blocking.

Lastly, a matching $M$ is \emph{$\alpha$-individually weakly (super) stable}\footnote{We call it individually ``weak/super'' stability since it coincides with weak/super stability for $\ell=1$.}
if for each unmatched pair~$\{a,a'\}\notin M$, there exists $b\neq b' \in \{a,a'\}$ such that there are $\alpha$ layers in which $b$ does not approve $b'$ or in which $b$ is happy (in which $b$ does not approve $b'$ and in which $b$ is happy).
Individual stability can also be interpreted as follows:
For each unmatched pair, at least one involved agent ``favors'' the current matching to the other agent from the pair in $\alpha$ layers (for weak stability, ``favors''
means prefers or is indifferent, for super stability,
``favors'' means prefers). 
By definition, every~$\alpha$-individually weakly/super stable matching is also~$\alpha$-pair weakly/super stable.
The subtle difference between pair and individual stability is that for each unmatched pair, one agent needs to prevent the pair from blocking in~$\alpha$~layers for individual stability, while the two agents together need to prevent the pair from blocking in~$\alpha$ layers for pair stability.
Note that this difference disappears for~$\alpha = 1$ ($1$-pair weak/super stability is equivalent to $1$-individual weak/super stability).  
In contrast, there are matchings that are all-layers weakly/super stable but not $\ell$-individually weakly/super stable.\footnote{For example, let $A=\{a_1,a_2,a_3,a_4\}$ and $\ell=2$ with $a_1$ and~$a_2$ approving each other in the first layer and~$a_3$ and~$a_4$ approving each other in the second layer. Then, $M=\{\{a_1,a_2\},\{a_3,a_4\}\}$ is all-layers super stable but not~$\ell$-individually super stable. Modifying the instance by letting~$a_1$ and~$a_3$ approve each other in both layers, $M$ is all-layers weakly stable but not $\ell$-individually weakly stable.}

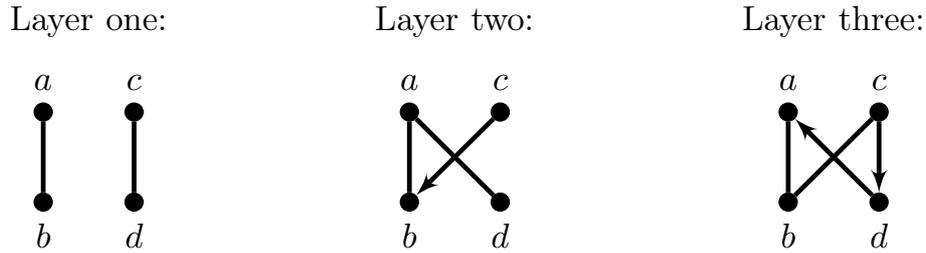
\begin{figure*}[t]
		\begin{center}
			\resizebox{!}{0.22\textwidth}{%
				\begin{tikzpicture}
				\tikzset{edge/.style = {->,> = latex'}}
				\node at (0.5, 1) {Layer one:};
				\node[vertex, label=90:$a$] (c1) at (0,0) {};
				\node[vertex, label=90:$c$] (c2) at ($(c1) + (1,0)$) {};
				
				\node[vertex, label=270:$b$] (c1h) at ($(c1) + (0, -1)$) {};
				\node[vertex, label=270:$d$] (c2h) at ($(c1h) + (1,0)$) {};
				
				\draw[line width=1.5pt] (c1) to (c1h);
				\draw[line width=1.5pt] (c2) to (c2h);
				\end{tikzpicture}}
			\qquad \qquad \qquad
			\resizebox{!}{0.22\textwidth}{%
				\begin{tikzpicture}
				\tikzset{edge/.style = {->,> = latex'}}
				\node at (0.5, 1) {Layer two:};
				\node[vertex, label=90:$a$] (c1) at (0,0) {};
				\node[vertex, label=90:$c$] (c2) at ($(c1) + (1,0)$) {};
				
				\node[vertex, label=270:$b$] (c1h) at ($(c1) + (0, -1)$) {};
				\node[vertex, label=270:$d$] (c2h) at ($(c1h) + (1,0)$) {};
				
				\draw[line width=1.5pt] (c1) to (c1h);
				\draw[line width=1.5pt] (c1) to (c2h);
				\draw[edge,line width=1.5pt] (c2) to (c1h);
				\end{tikzpicture}}
			\qquad \qquad \qquad
			\resizebox{!}{0.22\textwidth}{%
				\begin{tikzpicture}
				\tikzset{edge/.style = {->,> = latex'}}
				\node at (0.5, 1) {Layer three:};
				\node[vertex, label=90:$a$] (c1) at (0,0) {};
				\node[vertex, label=90:$c$] (c2) at ($(c1) + (1,0)$) {};
				
				\node[vertex, label=270:$b$] (c1h) at ($(c1) + (0, -1)$) {};
				\node[vertex, label=270:$d$] (c2h) at ($(c1h) + (1,0)$) {};
				
				\draw[line width=1.5pt] (c1) to (c1h);
				\draw[line width=1.5pt] (c2) to (c1h);
				\draw[edge,line width=1.5pt] (c2) to (c2h);
				\draw[edge,line width=1.5pt] (c2h) to (c1);
				\end{tikzpicture}}
		\end{center}
		\caption{Example with four agents having three-layered approval preferences over each other.
		An undirected edge between two agents means that they mutually approve each other in the respective layer, while a directed edge from an agent~$a^*$ to~$b^*$ means that~$a^*$~approves~$b^*$ (but~$b^*$ not~$a^*$).}
		\label{fig:example}
	\end{figure*}

\begin{example}
\Cref{fig:example} depicts an example with four agents and three layers. 
 For each layer, we depict the approval graph. 
 In the first layer, agent's approvals are symmetric but not in the other two layers. 
 The matching $M_1:=\{\{a,b\},\{c,d\}\}$ is all-layers weakly stable, as the agents~$a$ and~$b$ are happy in all layers. 
 Matching~$M_1$ is $2$-globally/pair strongly stable (because it is strongly stable in layers one and three) and $1$-globally super stable (because it is super stable in layer one). 
 Moreover, $M_1$ is $2$-pair/individually super stable (as the only blocking pairs under super stability are $\{a,d\}$ in layer two and $\{b,c\}$ in layer three). 
 
 The matching $M_2:=\{\{a,d\},\{b,c\}\}$ is $2$-globally/pair weakly stable (as it is weakly stable in layers two and three).
 However, $M_2$ is only $1$-individually weakly stable because of the pair $\{a,b\}$.
 Moreover, $M_2$ is not $\alpha$-globally/pair strongly stable and $\alpha$-globally/pair/individually stable for any $\alpha>0$, as the pair~$\{a,b\}$ blocks $M_2$ in all layers under strong and super stability.
\end{example}

Given a set $A$ of agents and their preferences~$(\T_a^i)_{a\in A,i\in [\ell]}$ in $\ell$ layers, \textsc{All-Layers Weak/Strong/Super Stability} is the problem of deciding whether there is an all-layers weakly/strongly/super stable matching.
In \textsc{Global Weak/Strong/Super Stability}, we are additionally given a parameter $\alpha\in [\ell]$ and the question is to decide whether there is an $\alpha$-globally weakly/strongly/super stable matching.
The problems \textsc{Pair Weak/Strong/Super Stability} and \textsc{Individual Weak/Super Stability} are defined analogously.

\section{Weak Stability} \label{se:weak}
A stable matching is guaranteed to exist in the single-layer weak stability setting: 
Only agents that approve each other can form a blocking pair. 
Thus, a maximal matching in~$G=(A,\{ \{a,a'\}\in \binom{A}{2}\mid a\text{ and } a' \text{ approve each other}\})$ is weakly stable.
It follows that a~$1$-globally/pair/individually weakly stable matching always exists and can be found in linear time. 
\paragraph{All-Layers Stability.}\label{sub:weak-alllayer}
We start by showing that as soon as we add a second layer, \textsc{All-Layers Weak Stability} becomes NP-hard even in very restrictive settings:
\begin{restatable}{theorem}{alllayerweak}
	\label{th:all-layer-weak}
	For each $\ell \ge 2$, \textsc{All-Layers Weak Stability} is NP-hard for symmetric bipartite approvals and NP-hard for (asymmetric) bipartite approvals even if agents from one side approve the same agents in both layers. 
	Both results hold even if each agent approves at most three agents in each layer.
\end{restatable}

\begin{proof}
  We first show the NP-hardness for~$\ell = 2$ and afterwards derive hardness for $\ell > 3$.
  
	First we show the first part of the theorem by reducing from the NP-hard variant of \textsc{Satisfiability} where each clause consists of exactly three literals and each variable occurs positively in at most two clauses and negatively in at most two clauses \cite{DBLP:journals/eccc/ECCC-TR03-049}.  
	
	Let $(X,C)$ be an instance of the above described variant of \textsc{Satisfiability} where $X$ is the set of variables and $C$ the set of clauses. 
	We construct an instance of \textsc{All-Layers Weak Stability} with two layers and symmetric approvals.
	For each variable $x\in X$ we introduce a variable gadget consisting of four agents: $a_{x}$, $a_{\bar{x}}$, $b^+_{x}$, $b^-_{x}$. 
	In the first layer, $a_x$ and $a_{\bar{x}}$ approve both $b^+_{x}$ and~$b^-_{x}$. 
	In the second layer, $a_x$ and $a_{\bar{x}}$ approve~$b^+_{x}$. 
	Note that matching~$a_x$ to $b^+_{x}$ will correspond to setting $x$ to true, while matching $a_{\bar{x}}$ to~$b^+_{x}$ will correspond to setting $x$ to false. 
	For each clause~$c=z^1\vee z^2 \vee z^3\in C$, we introduce a clause gadget consisting of five agents~$\alpha_c^1$,~$\alpha_c^2$,~$\beta_c^1$,~$\beta_c^2$, $\beta_c^3$. 
	In both layers, agent $\alpha_c^1$ approves agents $\beta_c^1$ and $\beta_c^2$ and agent $\alpha_c^2$ approves agents $\beta_c^2$ and $\beta_c^3$. 
	Moreover, for $i\in [3]$, in the second layer, $\beta_c^i$ approves~$a_{z^i}$. 
	We complete all approvals such that they are symmetric.
	See \Cref{fig:all-layers-weak} for an illustration.

	All agents approve at most two agents in the first layer and at most three agents in the second layer (as each literal appears in at most two clauses).
	Moreover, approvals are clearly bipartite with agents~$a_x$, $a_{\bar x}$ for $x \in X$ and $\alpha_c^1$, $\alpha_c^2$ for~$c \in C$ on the one and $b_x^+$, $b_x^-$ for $ x\in X$ and $\beta_c^1$, $\beta_c^2$, and $\beta_c^3$ for~$c \in C$ on the other side.
	
		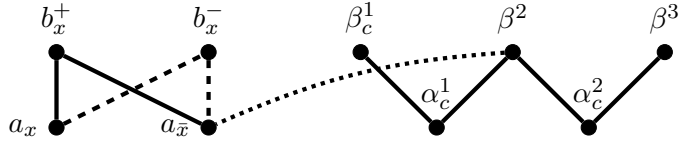
\begin{figure}[t]
		\begin{center}
			\begin{tikzpicture}
			\node[vertex, label=180:$a_x$] (xtrue) at (0,0) {};
			\node[vertex, label=180:$a_{\bar x}$] (xfalse) at ($(xtrue) + (2,0)$) {};
			
			\node[vertex, label=90:$b^+_x$] (xf) at ($(xtrue) + (0, +1)$) {};
			\node[vertex, label=90:$b^-_x$] (xd) at ($(xf) + (2,0)$) {};
			
			\draw (xtrue) edge[ultra thick, dashed] (xd);
			\draw (xfalse) edge[ultra thick] (xf);
			
			\draw (xf) edge[ultra thick] (xtrue);
			\draw (xd) edge[ultra thick, dashed] (xfalse);
			
			\node[vertex, label=90:$\beta_c^1$] (c1) at ($(xf) + (4, 0)$) {};
			\node[vertex, label=90:$\beta^2$] (c2) at ($(c1) + (2,0)$) {};
			\node[vertex, label=90:$\beta^3$] (c3) at ($(c2) + (2,0)$) {};
			
			\node[vertex, label=90:$\alpha_c^1$] (c1h) at ($(c1) + (1, -1)$) {};
			\node[vertex, label=90:$\alpha_c^2$] (c2h) at ($(c1h) + (2,0)$) {};
			
			\draw (c2) edge[ultra thick] (c1h);
			
			\draw (c1h) edge[ultra thick] (c1);
			\draw (c2h) edge[ultra thick] (c2);
			
			\draw (c3) edge[ultra thick] (c2h);
			
			\draw (c2) edge[ultra thick, dotted, bend right=10] (xfalse);
			\end{tikzpicture}
			
		\end{center}
		\caption{Visualization of construction from \Cref{th:all-layer-weak}. 
		We depict the agents introduced for a variable~$x$ and a clause~$c$, where~$\bar x$ is the second literal in~$c$.
        Approvals appearing only in layer~1/2 are dashed/dotted. 
        Approvals appearing in both layers are solid.
		}
		\label{fig:all-layers-weak}
	\end{figure}
	
	{\bfseries ($\Rightarrow$)}
	Let $X'$ be the set of variables that are set to true in a satisfying assignment of $(X,C)$.
	We construct a matching~$M$ in our instance as follows.
	For each~$x\in X'$, we match~$a_x$ with~$b^+_x$ and~$a_{\bar{x}}$~with~$b^-_x$; for each $x\notin X'$, we match~$a_{\bar{x}}$ with $b^+_x$ and $a_{x}$ with $b^-_x$. 
	Moreover, for clause~$c\in C$, let $z_c^i$ for~${i\in [3]}$ be a literal from $c$ that is fulfilled (such a literal exists as we consider a satisfying assignment). 
	We construct a matching inside the clause gadget such that $\beta_c^i$ is unmatched while~$\beta_c^j$~for~$j\in [3]\setminus \{i\}$ are happy in both layers and thus cannot be part of a blocking pair.
	If $i=1$, then we match $\alpha_c^1$ to $\beta_c^2$ and $\alpha_c^2$ to $\beta_c^3$;
	if $i=2$, then we match~$\alpha_c^1$ to $\beta_c^1$ and $\alpha_c^2$ to $\beta_c^3$;
	if $i=3$, then we match $\alpha_c^1$ to $\beta_c^1$ and $\alpha_c^2$ to~$\beta_c^2$.
	As also $\alpha_c^1$ and $\alpha_c^2$ are happy in both layers, agent~$\beta_c^i$ can only form a blocking pair with~$a_{z_c^i}$. 
	Since we started with a satisfying assignment and by the choice of~$i$, $a_{z_c^i}$ is happy in both layers and thus cannot be part of a blocking pair (as it is matched to $b^+_x$).
	Thus, there cannot be a blocking pair involving an agent from a clause gadget and there is also clearly no blocking pair inside a single variable gadget.
	
	{\bfseries ($\Leftarrow$)}
	Assume we are given a matching~$M$ that is weakly stable in both layers. 
	First we show that for each variable~$x\in X$, agent~$a_x$ and $a_{\bar{x}}$ need to be matched to $b^+_x$ or $b^-_x$:
	Assume for the sake of contradiction that without loss of generality $a_x$ is matched to neither~$b^+_x$ nor $b^-_x$. 
	As~$a_x$ only approves~$b^+_x$ and~$b^-_x$ in the first layer, $a_x$ is unhappy in the first layer.
	Moreover, note that~$b^+_x$ and~$b^-_x$ only approve $a_x$ and $a_{\bar{x}}$ in the first layer, implying that one of these two agents is unhappy in the first layer in this case and thus forms a blocking pair together with~$a_x$ in the first layer, contradicting the stability of~$M$. 
	
	Thus, all agents from a variable gadget are matched inside the gadget.
	By analogous arguments, we have that agents~$\alpha_c^1$ and $\alpha_c^2$ are matched to $\beta_c^1$, $\beta_c^2$, or $\beta_c^3$.
	Hence, for each clause~$c\in C$, there is some~$i\in [3]$ such that $\beta_c^{i}$ is unhappy in both layers, as taking into account only agents from clause gadgets, $\beta_c^1$, $\beta_c^2$, $\beta_c^3$ all only approve $\alpha_c^1$ and~$\alpha_c^2$ in some layer.
	Now, let $X'\subseteq X$ be the subset of variables~$x\in X$ where $a_x$ is matched to $b^+_x$. 
	We claim that setting all variables from $X'$ to true and all variables from $X\setminus X'$ to false results in a satisfying assignment.
	Assume for the sake of contradiction that there is a clause $c = z^1 \lor z^2 \lor z^3 \in C$ which is not satisfied.
	Let~$i\in [3]$ such that $\beta_c^{i}$ is unhappy in both layers, as taking into account only agents from clause gadget
	This implies that $z^{i}$ is not satisfied which by the construction of $X'$ implies that $a_{z^{i}}$ is matched to $b^-_{x}$.
	However, from this it follows that both $a_{z^{i}}$ and $\beta_c^{i}$ are unhappy in the second layer and thus form a blocking pair, contradicting the stability of~$M$. 
	
	\medskip
	For the second part of the statement note that a pair~$\{a, b\}$ can only be blocking if both~$a$ and~$b$ approve each other and thus, the construction also goes trough if all $b$-agents that currently approve an $a$-agent only in one layer also approve it in the other layer.
	
	\medskip
	To show hardness for~$\ell >3$, we add $\ell -3$ empty layers (i.e., layers in which no agent approves any other agent) in the case of symmetric preferences.
	In the case of asymmetric preferences where all agents from one side, say~$B$, have the same preferences, we add~$\ell -3$ layers in which only agents from~$B$ approve other agents.
	As every matching is stable in the additional~$\ell -3$~layers, the NP-hardness follows.	
\end{proof}

Since a matching is weakly stable in the case of symmetric approvals in some layer $i\in [\ell]$ if and only if it is a maximal matching in~${G_i=(A,\{ \{a,a'\}\mid a\in T_{a'}^i, a'\in T_a^i\})}$, this also shows that finding a matching that is maximal in both layers of a two-layered graph is NP-hard.
This result might be of independent interest.

\paragraph{Global Stability.}\label{sub:weak-global}
The NP-hardness of \textsc{Global Weak Stability} for $\alpha=\ell \ge 2$ already follows from \Cref{th:all-layer-weak}.
We now analyze the problem's complexity for other values of $\alpha$: 
We have observed that \textsc{Global Weak Stability} is linear-time solvable for $\alpha=1$. 
However, by reducing from \textsc{All-Layers Weak Stability} and adding layers where all matchings are stable or a set of ``conflicting'' layers where a matching can be only stable in one of them, we show~NP-hardness for all other values of $\alpha$: 
\begin{restatable}{proposition}{globalweak}
	\label{pr:global-weak}
	For any $2\leq \alpha\leq \ell$, \textsc{Global Weak Stability} is NP-hard for symmetric bipartite approvals. 
\end{restatable}
\begin{proof}
	We reduce from \textsc{All-Layers Weak Stability} with two layers, which is NP-hard as proven in \Cref{th:all-layer-weak}.
	Given an instance $(A,(\T_a^i)_{a\in A,i\in [2]})$, we construct an instance of \textsc{Global Weak Stability}  with arbitrary~${2\leq \alpha\leq \ell}$.
	The set of agents is~${A\cup \{a^*,b^*\}\cup \{c_i \mid i\in [3,\ell-\alpha+2]\}}$. 
	Note that the set $\{c_i \mid i\in [3,\ell-\alpha+2]\}$ might be empty.
	In the first two layers agents from $A$ have their preferences as described in $(\T_a^i)_{a\in A,i\in [2]}$ while they do not approve any agent in layers three to $\ell$. 
	Agent~$a^*$ approves $b^*$ in the first two layers and agent $c_i$ in layer $i\in [3,\ell-\alpha+2]$ while it does not approve any agent in layers $\ell-\alpha+3$ to $\ell$. 
	We complete approvals to be symmetric. 
	
	{\bfseries ($\Rightarrow$)} 
	Let $M$ be an all-layers stable matching in~$(A,(\T_a^i)_{a\in A,i\in [2]})$.
	Then, we extend $M$ by matching~$a^*$ to $b^*$. 
	Note that the resulting matching $M$ is still stable in the first and second layer. 
	Moreover, as there are no approvals in layers $\ell-\alpha+3$ to $\ell$ matching $M$ is also stable in these layers. 
	Thus, $M$ is stable in $\alpha$ layers. 
	
	{\bfseries ($\Leftarrow$)}
	Assume that there is a matching $M$ that is stable in a subset $S\subseteq [\ell]$ of $\alpha$ layers in the constructed instance. 
	For the sake of contradiction assume that there is some ${i\in [3,\ell-\alpha+2]}$ with $i\in S$. 
	As $a^*$ and $c_i$ approve each other and no other agent in layer $i$, matching~$M$ needs to contain~$\{a^*,c_i\}$. 
	However, from this it follows that $M$ is unstable in the first and second layer, as $b^*$ and $a^*$ approve each other and no other agent in these layers. 
	Moreover, $M$ is unstable in layer $j\in [3,\ell-\alpha+2]\setminus \{i\}$, as $a^*$ and $c_j$ form a blocking pair. 
	This implies that $M$ is stable in at most~$\alpha-1$~layers, a contradiction. 
	Thus, $M$ can only be stable in layers~$[2]\cup [\ell-\alpha+3,\ell]$. 
	As these are $\alpha$ layers, it follows that $M$ is in particular stable in the first two layers. 
	Hence, $M$ restricted to the agents from $A$ is a stable matching in both layers in $(A,(\T_a^i)_{a\in A,i\in [2]})$. 
\end{proof}

\paragraph{Pair and Individual Stability.} \label{sub:weak-pair} For $\alpha=\ell \ge 2$, \Cref{th:all-layer-weak} shows NP-hardness of \textsc{Pair Weak Stability} and the same construction also proves NP-hardness of \textsc{Individual Weak Stability}.
However, trying to extend this result to any~$\alpha$ and $\ell$, the idea of simply adding ``conflicting'' layers here fails, because the ``stable'' layers can be different for each agent/pair. 
Nevertheless, reducing from \textsc{Pair/Individual Weak Stability} for $\alpha=\ell=2$ and having~$\lceil \ell / 2 \rceil$ copies of the first and~$\lfloor \ell / 2 \rfloor$ copies of the second layer, NP-hardness for any~$\ell\ge 2$ and~$\alpha>\lceil \ell / 2 \rceil$ follows.
In contrast to this, for $\alpha \le \lceil \ell / 2 \rceil$, an $\alpha$-individually/pair weakly stable matching always exists: 
\begin{restatable}{theorem}{individualweak}
	\label{th:weak-individual} \label{th:pair-weakP} \label{weak:pair-NP}
	For any~$\ell \ge 2$ and $\alpha \le \lceil \ell / 2 \rceil$, an $\alpha$-pair/individually weakly stable matching always exists and can be found in linear time.
	For any~$\ell\ge 2$ and $\alpha>\lceil \ell / 2 \rceil$,  \textsc{Pair/Individual Weak Stability} are NP-hard for symmetric bipartite approvals even if each agent approves at most three agents in each layer.  
\end{restatable}
\begin{proof}
    \noindent$\bm{\alpha \le \lceil \ell / 2 \rceil.}$
	We construct a graph~$G$ with vertex set~$A$ and an edge~$\{a, a'\}\in \binom{A}{2}$ if and only if there are at least~$\ell - \alpha + 1 $ layers in which~$a$~approves $a'$ and there are at least~$\ell - \alpha + 1$ layers in which $a'$ approves $a$.
	Let $M$ be a maximal matching in~$G$.
	We claim that $M$ is $\alpha$-individually weakly stable.
	Consider a pair~$\{a, a'\}$.
	If~$\{a, a' \} \notin E(G)$, then $a$ does not approve~$a'$ in at least~$\alpha$ layers or~$a'$ does not approve~$a$ in at least~$\alpha$ layers.
	It follows that this pair is not blocking.
	Otherwise it follows by the maximality of~$M$ that $a$ or $a'$ is matched by~$M$;
	we assume without loss of generality that~$a$ is matched by~$M$.
	Then, $a$ is happy in~$M$ in at least~$\ell - \alpha + 1 \ge \alpha$ layers, implying that~$a$ does not prefer~$a'$ to~$M(a)$ in at least $\alpha$ layers.
	Thus, $M$ is $\alpha$-individually weakly stable and thereby also $\alpha$-pair weakly stable. \medskip
	
	\noindent$\bm{\alpha > \lceil \ell / 2 \rceil.}$ We start by proving hardness for \textsc{Pair Weak Stability}: 
	We reduce from \textsc{All-Layers Weak Stability} for two layers which is NP-hard for symmetric and bipartite approvals where each agent approves at most three agents in each layer as proven in \Cref{th:all-layer-weak}.
	Given an instance~$\mathcal{I}'=(A,(\T_a^i)_{a\in A,i\in [2]})$ of \textsc{All-Layers Weak Stability}, we construct an equivalent instance $\mathcal{I}$ of \textsc{Pair Weak Stability} for some $\ell\ge 2$ and $\alpha>\lceil \ell / 2 \rceil$ by replacing the first layer by $\lceil \ell / 2 \rceil$ many copies and the second layer by $\lfloor \ell / 2 \rfloor$ copies. 
	
	{\bfseries ($\Rightarrow$)} 
	Assuming that there is an all-layers weakly stable matching in $\mathcal{I}'$, then this matching is $\ell$-pair and thus $\alpha$-pair weakly stable in $\mathcal{I}$. 
	
	{\bfseries ($\Leftarrow$)}
	Assume that there is a $\alpha$-pair weak stable matching~$M$ in $\mathcal{I}$.
	Then, in $M$ there is no pair that is weakly blocking in~${\ell-\alpha+1}$~layers. 
	As $\alpha>\lceil \ell / 2 \rceil$, it follows that there is no pair that is blocking in  $\lfloor \ell / 2 \rfloor$ layers. 
	However, as there are only two different layers with each of them occurring at least~$\lfloor \ell / 2 \rfloor$~times, a pair that blocks a single layer automatically blocks at least~$\lfloor \ell / 2 \rfloor$~layers. 
	Thus, there is no blocking pair for $M$ in any layer implying that $M$ is all-layers weakly stable and thus a solution to $\mathcal{I}'$.
	
	We now turn to \textsc{Individual Weak Stability}: 
	Note that our construction from \Cref{th:all-layer-weak} also proves hardness for finding a $2$-individually weakly stable matching in two layers: 
	As an~$\ell$-individually stable matching is also all-layers stable, it remains to check whether the matching~$M$ constructed in the forward direction of the proof is $2$-individually weakly stable: 
	For each unmatched pair of agents that approve each other in both layers, one agent is happy in both layers, while for each unmatched pair that approve each other in one layer, as the pair is not weakly blocking in this layer, one agent needs to be happy in this layer. From this it follows that $M$ is $2$-individually weakly stable. 
	To extend the hardness result for any~$\ell\ge 2$ and~$\alpha>\lceil \ell / 2 \rceil$ it is now possible to apply the same reduction and arguments as the ones for pair stability from above.
\end{proof}
\Cref{th:weak-individual} indicates that individual and pair stability lead to similar complexity results for weak stability, and that a low degree
of stability  leads to tractability for these two notions.
Interestingly, the latter is in stark contrast to the results of Chen et al. \cite{DBLP:conf/sigecom/ChenNS18}, who showed that the problem of finding pair or individually stable matchings is polynomial-time solvable for~$\alpha > \lfloor \ell / 2 \rfloor$ and NP-hard for~${\alpha \leq \lfloor \ell / 2 \rfloor}$ for strict preferences under certain constraints (note that we have similar results for individual/pair super stability, see \Cref{th:super-pair-neg,th:super-pair-pos}).

\section{Strong Stability}  \label{se:strong}
In the single-layer setting, a strongly stable matching may fail to exist (consider as an example three agents all approving each other); however, deciding the existence of a strongly stable matching is polynomial-time solvable \cite{DBLP:conf/esa/Kunysz16}.
\paragraph{All-Layers Stability.} \label{sub:strong-alllayer}
In contrast to weak stability, for symmetric approvals \textsc{All-Layers Strong Stability} is polynomial-time solvable. 
\begin{proposition} \label{th:strong-alllayer}
	\textsc{All-Layers Strong Stability} for symmetric approvals can be solved in~$O(\ell n^2 + n^{2.5})$~time.
\end{proposition}
\begin{proof}
	Recall that for some $i\in [\ell]$, ${G_i=(A,E_i=\{ \{a,a'\}\mid a\in T_{a'}^i, a'\in T_a^i\})}$.
	First assume that~$n$~is even.
	We claim that there is an all-layers strongly stable matching if and only if there is a perfect matching in the undirected graph~${H = (A, \bigcap_{i \in [\ell]} E_i')}$, where 
	$$E_i' := E_i \cup \{ \{ a, a' \} \mid a \text{ and } a' \text{ do not approve any agent in layer }i \}.$$
	Note that an edge~$e = \{ a, a' \}$ is present in $H$ if for every~$i \in [\ell]$, either $e \in E_i$ or both $a$ and $a'$ have no neighbor~in~$G_i$.
	 
	Let~$M$ be a perfect matching in~$H$.
	Assume towards a contradiction that there is a blocking pair~$\{a, a'\}$ in some layer~$i$.
	Then $a$ and $a'$ approve each other in layer~$i$.
	Consequently, both~$a$ and~$a'$ have a neighbor in $G_i$ and thus approve all their neighbors in $(A,E'_i)$.
	As $M$ is a perfect matching in $(A,E'_i)$, it follows that $a$ approves~$M(a)$ and $a'$ approves $M(a')$ in layer~$i$, a contradiction to $\{a, a'\}$ being blocking.
	
	Now assume that $H$ does not admit a perfect matching.
	Let~$M$ be any matching.
	Without loss of generality, we assume that $M$ matches all agents (note that such a matching $M$ is not a perfect matching in $H$ because $M$ might include edges/pairs that are not present in $H$).
	We can assume this, as a matching can never become unstable by adding further agent pairs as agents are indifferent between being unmatched and being matched to agents they do not approve. 
	Since~$M$ is not a perfect matching in $H$, matching~$M$ contains an edge~$\{a, a'\} \notin E_i'$ for some some~$i \in [\ell]$.
	Consequently, $a$ and $a'$ do not approve each other in layer~$i$, and one of~$a$ and $a'$ (without loss of generality~$a$) approves some other agent~$b$ in layer~$i$.
	As approvals are symmetric, it follows that~$\{a, b\}$ blocks~$M$ in layer~$i$, implying that~$M$ is not all-layers strongly stable.
	
	If $n$ is odd, each matching $M$ leaves at least one agent~${a\in A}$ unmatched. 
	If~$a$ approves some agent~$a'$ in layer~$i\in [\ell]$, then $\{a, a'\}$ blocks $M$ in layer $i$.
	Thus, in an instance with an odd number of agents and a stable matching, there must be an agent $a$ which, in each layer, does not approve any agent. 
	If the instance does not contain such an agent, then we return that there is no stable matching.
	Otherwise we pick an agent~$a$ approving no other agent and check
	whether the instance arising from the deletion of~$a$ (which has an even number of agents) has an all-layers strongly stable matching. 
\end{proof}

For asymmetric approvals, \textsc{All-Layers Strong Stability} becomes NP-hard. 
This is one of only few cases where asymmetric and symmetric approvals computationally differ.
It is open whether \textsc{All-Layers Strong Stability} for asymmetric (bipartite) approvals is NP-hard for two layers.

\begin{restatable}{theorem}{alllayerstrong}
	\label{thm:all-layer-strong}
	\textsc{All-Layers Strong Stability} for bipartite asymmetric approvals is NP-hard for $\ell \geq 3$. 
\end{restatable}
\begin{proof}
	We reduce from \textsc{Monotone 3-Sat}, the restriction of \textsc{3-Sat} where each clause contains either only negated or only non-negated literals~\cite{DBLP:conf/stoc/Schaefer78}.
	We first construct the first three layers of the instance of \textsc{All-Layers Strong Stability}.
	Afterwards, we add~$\ell - 3$ empty layers (in which every matching is stable).
	Thus, a matching is weakly stable in all layers if and only if it is stable in the first three layers.
	
	For each variable~$x$, we add four agents~$x_0, x_1, x_2$, and $x_3$.
	In each of the first three layers, $x_i$ approves $x_{i+1}$ for $i\in \{0,1,2,3\}$ (where $i+1$ is taken modulo $4$).
	
	For each clause~$c$, we add six agents~$c_0$, $c_1$, $c_2$, $\hat c_0$, $\hat c_1$, and~$\hat c_2$.
	In layer one, for each $i\in \{0, 1, 2\}$, $c_i$ approves $\hat c_i$, while $\hat c_i$ approves $c_j$ for all~$j\in \{0,1,2\}\setminus \{i\}$.
	In layer two, for each $i\in \{0,1,2\}$, $c_i$ approves $\hat c_{i+1}$, while $\hat c_i$ approves $c_j$ for all~$j\in \{0,1,2\} \setminus \{i+1\}$ (where $i+1$ is taken modulo 3).
	In layer three, for each~$i\in \{0,1,2\}$, $c_i$ approves $\hat c_{i-1}$, while $\hat c_i$ approves~$c_j$ for all $j\in \{0,1,2\}\setminus \{i-1\}$ (where $i-1$ is taken modulo 3).
	See \Cref{fig:clause-gadget-all-layer-strong} for an example.
	Notably, there are no agents~$c_i$ and~$\hat c_j$ approving each other in the same layer.
	\begin{figure*}[t]
		\begin{center}
			\resizebox{!}{0.27\textwidth}{%
				\begin{tikzpicture}
				\node at (-1, 1) {Layer one:};
				\node[vertex, label=90:$c_0$] (c1) at (0,0) {};
				\node[vertex, label=90:$c_1$] (c2) at ($(c1) + (2,0)$) {};
				\node[vertex, label=90:$c_2$] (c3) at ($(c2) + (2,0)$) {};
				
				\node[vertex, label=270:$\hat c_0$] (c1h) at ($(c1) + (0, -2)$) {};
				\node[vertex, label=270:$\hat c_1$] (c2h) at ($(c1h) + (2,0)$) {};
				\node[vertex, label=270:$\hat{c}_2$] (c3h) at ($(c2h) + (2,0)$) {};
				
				\draw (c1) edge[approves, thick, postaction={decorate}] (c1h);
				\draw (c2) edge[approves, thick, postaction={decorate}] (c2h);
				\draw (c3) edge[approves, thick, postaction={decorate}] (c3h);
				
				\draw (c1h) edge[approves, thick, postaction={decorate}] (c2);
				\draw (c2h) edge[approves, thick, postaction={decorate}] (c3);
				\draw (c3h) edge[approves, thick, postaction={decorate}] (c1);
				
				\draw (c1h) edge[approves, thick, postaction={decorate}] (c3);
				\draw (c2h) edge[approves, thick, postaction={decorate}] (c1);
				\draw (c3h) edge[approves, thick, postaction={decorate}] (c2);
				\end{tikzpicture}}
			\qquad \qquad
			\resizebox{!}{0.27\textwidth}{%
				\begin{tikzpicture}
				\node at (-1, 1) {Layer two:};
				\node[vertex, label=90:$c_0$] (c1) at (0,0) {};
				\node[vertex, label=90:$c_1$] (c2) at ($(c1) + (2,0)$) {};
				\node[vertex, label=90:$c_2$] (c3) at ($(c2) + (2,0)$) {};
				
				\node[vertex, label=270:$\hat c_0$] (c1h) at ($(c1) + (0, -2)$) {};
				\node[vertex, label=270:$\hat c_1$] (c2h) at ($(c1h) + (2,0)$) {};
				\node[vertex, label=270:$\hat{c}_2$] (c3h) at ($(c2h) + (2,0)$) {};
				
				\draw (c1) edge[approves, thick, postaction={decorate}] (c2h);
				\draw (c2) edge[approves, thick, postaction={decorate}] (c3h);
				\draw (c3) edge[approves, thick, postaction={decorate}] (c1h);
				
				\draw (c1h) edge[approves, thick, postaction={decorate}] (c2);
				\draw (c2h) edge[approves, thick, postaction={decorate}] (c3);
				\draw (c3h) edge[approves, thick, postaction={decorate}] (c1);
				
				\draw (c1h) edge[approves, thick, postaction={decorate}] (c1);
				\draw (c2h) edge[approves, thick, postaction={decorate}] (c2);
				\draw (c3h) edge[approves, thick, postaction={decorate}] (c3);
				\end{tikzpicture}}
			\qquad \qquad
			\resizebox{!}{0.27\textwidth}{%
				\begin{tikzpicture}
				\node at (-1, 1) {Layer three:};
				\node[vertex, label=90:$c_0$] (c1) at (0,0) {};
				\node[vertex, label=90:$c_1$] (c2) at ($(c1) + (2,0)$) {};
				\node[vertex, label=90:$c_2$] (c3) at ($(c2) + (2,0)$) {};
				
				\node[vertex, label=270:$\hat c_0$] (c1h) at ($(c1) + (0, -2)$) {};
				\node[vertex, label=270:$\hat c_1$] (c2h) at ($(c1h) + (2,0)$) {};
				\node[vertex, label=270:$\hat{c}_2$] (c3h) at ($(c2h) + (2,0)$) {};
				
				\draw (c1) edge[approves, thick, postaction={decorate}] (c3h);
				\draw (c2) edge[approves, thick, postaction={decorate}] (c1h);
				\draw (c3) edge[approves, thick, postaction={decorate}] (c2h);
				
				\draw (c1h) edge[approves, thick, postaction={decorate}] (c3);
				\draw (c2h) edge[approves, thick, postaction={decorate}] (c1);
				\draw (c3h) edge[approves, thick, postaction={decorate}] (c2);
				
				\draw (c1h) edge[approves, thick, postaction={decorate}] (c1);
				\draw (c2h) edge[approves, thick, postaction={decorate}] (c2);
				\draw (c3h) edge[approves, thick, postaction={decorate}] (c3);
				\end{tikzpicture}}
			
		\end{center}
		\caption{Agents added for a clause~$c$ and their approvals in the three layers.
			An arrow from an agent~$a$ to an agent~$b$ indicates that $a$ approves~$b$.}
		\label{fig:clause-gadget-all-layer-strong}
	\end{figure*}
	
	Let $z$ be the $i$th literal of clause~$c$.
	If $z = x$, then $x_0$ approves $\hat c_i$ in layer $i$.
	If~$z = \bar{x}$, then $x_1$ approves $\hat c_i$ in layer~$i$.
	This finishes the construction.
	
	We next show that the approvals are bipartite.
	To this end, let~$X$ be the set of variables, let~$\mathcal{C}^{\text{pos}}$ be the set of clauses which only contain non-negated literals and let $\mathcal{C}^\text{neg}$ be the set of clauses which only contain negated literals.
	Then agents from~${\{x_0, x_2 \mid x \in X\} \cup \{c_0, c_1, c_2 \mid c \in \mathcal{C}^\text{pos}\}\cup \{\hat c_0, \hat c_1, \hat c_2 \mid c \in \mathcal{C}^\text{neg}\}}$ only approve agents from~$\{x_1, x_3 \mid x \in X\}\cup \{c_0, c_1, c_2 \mid c \in \mathcal{C}^\text{neg}\} \cup \{\hat c_0, \hat c_1, \hat c_2 \mid c \in \mathcal{C}^\text{pos}\}$ and vice versa.
	
	$(\Rightarrow)$
	Assume that there is a satisfying assignment~$f$.
	We construct a stable matching~$M$ as follows.
	For each variable~$x$ which is set to true by~$f$, we add pairs~$\{x_0, x_1\}$ and $\{x_2, x_3\}$.
	For each variable~$x$ which is set to false by~$f$, we add pairs~$\{x_0, x_3\}$ and $\{x_1, x_2\}$.
	
	For each clause~$c$, fix some~$i_c\in \{0,1,2\}$ such that $c$ is satisfied by the $i_c$th literal of~$c$ in $f$.
	Add pairs~$\{c_j, \hat c_{j + i_c}\}$ for every~$j \in [3] $ (where $j + i_c$ is taken modulo 3).
	
	It remains to show that~$M$ is strongly stable in the first three layers.
	Clearly, no blocking pair is of the form~$\{x_i, y_j\}$ for variables~$x, y$ and ${i,j\in  \{0,1,2, 3\}}$.
	For each clause~$c$ and each layer~$i \in  \{0,1,2\}$, it holds that either agents $c_j$ for all~$j\in \{0,1,2\}$~are happy or agents~$\hat c_j$~for all~$j\in  \{0,1,2\}$~are happy.
	As no two agents $c_{j}$ and~$\hat c_{j'}$ approve each other in the same layer, there is no blocking pair of the form~$\{c_j, \hat c_{j'}\}$.
	The only remaining approvals are $x_0$ or~$x_1$ approving~$\hat c_i$ in layer $i$ where $x$ or $\bar x$ is the $i$th literal of~$c$.
	If~$i \neq i_c$, then $\hat c_i$ is happy in layer~$i$ and so neither~$\{x_0, \hat c_i\} $ nor~$\{x_1, \hat c_i\}$ is blocking.
	Otherwise the $i$th literal~$z$ of~$c$ satisfies~$c$.
	If $z = x$, then~$x_0$~is happy in each layer by the construction of~$M$ and thus $\{x_0, \hat c_i\}$ does not block~$M$.
	If $z = \hat x$, then $x_1$ is happy in each layer by the construction of~$M$ and thus $\{x_1, \hat c_i\}$ does not block~$M$.
	Consequently, $M$ is strongly stable in each layer.
	
	$(\Leftarrow)$
	Assume that there is a stable matching~$M$.
	First note that for each clause~$c$, we have~$\{ M(c_0), M(c_1), M(c_2) \} = \{ \hat c_0, \hat c_1, \hat c_2 \}$.
	If not, then there is one agent $c_j$ which is not matched to any of $\hat c_0, \hat c_1, \hat c_2$ and another agent $\hat c_{j'}$ which is not matched to any of $c_0, c_1, c_2$.
	Then,~$\{c_{j}, \hat c_{j'} \}$~blocks~$M$ in the first three layers, contradicting the stability of~$M$.
	In fact, we claim that there exists some~$i \in  \{0,1,2\}$ such that for every $j \in  \{0,1,2\}$, matching~$M$~contains~$\{ c_j, \hat c_{j + i}\}$ (where~$j + i$ is taken modulo 3).
	Otherwise, there exist some $j\neq j' \in  \{0,1,2\}$ such that~$\{c_j, \hat c_{j'}\} \in M$ and $\{c_{j'}, \hat c_j\} \in M$.
	Note however that $\{ c_j, \hat c_j \}$ blocks $M$ in the layer where $c_{j'}$ approves~$\hat c_j$, contradicting the stability of~$M$.
	Thus, our claim holds.
	
	We now show that $M$ contains either~$\{x_0, x_1\}$ and $\{x_2, x_3\}$ or $\{x_0, x_3\} $ and~$\{x_1, x_2\}$ for each variable~$x$.
	Since $x_0$ approves $x_1$ in the first three layers, $x_0$ or $x_1$ must be happy in each of these layers.
	It follows that $M$ contains either~$\{ x_0, x_1 \}$ or~$\{ x_1, x_2 \}$.
	It is easy to verify that~$\{ x_2, x_3 \} \in M$ if $\{ x_0, x_1 \} \in M$ and that $\{ x_3, x_0 \} \in M$ if $\{ x_1, x_2 \} \in M$.
	We define a satisfying assignment~$f$ by setting variable~$x$ to true if $M$ contains~$\{x_0, x_1\}$ and $\{x_2, x_3\}$ and false otherwise.
	
	It remains to show that each clause is satisfied.
	Fix a clause~$c$ and assume that for each~${j\in \{0,1,2\}}$, agent~$c_j$ is matched to~$\hat c_{j + i}$ (where $j +i$ is taken modulo~3).
	We claim that the $i$th literal~$z $ of $c$ satisfies~$c$.
	Note that $\hat c_i$ is unhappy in layer~$i$.
	If $z = x$, then for~$\{x_0, \hat c_i\}$~not to be blocking in layer~$i$, agent~$x_0$ needs to be happy in layer $i$ in~$M$.
	Consequently, $M$ contains pairs~$\{x_0, x_1\}$ and $\{x_2, x_3\}$ and $x$ is set to true by $f$.
	If $z = \bar x$, then for~$\{x_1, \hat c_i\}$ not to be blocking, $x_1$ needs to be happy in~$M$.
	Consequently, $M$ contains pairs~$\{x_0, x_3\}$ and $\{x_1, x_2\}$ and~$x$~is set to false by $f$.
	It follows that $c$ is satisfied. 
\end{proof}

\paragraph{Global Stability.}  \label{sub:strong-global}
\Cref{thm:all-layer-strong} also implies that \textsc{Global Strong Stability} for asymmetric approvals is NP-hard for all $3 \le \alpha \le \ell$.
This can be shown by adding an appropriate number of layers in which every matching is stable (e.\,g.\ layers without any approvals) and layers without any stable matching.

For symmetric approvals, on the other hand, we can solve \textsc{Global Strong Stability} in~${\ell \choose \alpha}\cdot n^{O (1)}$~time: For each subset~$S$ of~$[\ell]$ of size~$\alpha$, we check whether there is a matching stable in all layers of $S$  using \Cref{th:strong-alllayer}.
This results in the following.
\begin{corollary}\label{co:strong-globalPara}
	\textsc{Global Strong Stability} for symmetric approvals is in FPT wrt.~$\ell$, in XP wrt.~$\alpha$, and in XP wrt.~$\ell-\alpha$.\footnote{A problem is in FPT with respect to a parameter $k$ if it is solvable in $f(k)|\mathcal{I}|^{\mathcal{O}(1)}$ time and in XP if it solvable in $|\mathcal{I}|^{f'(k)}$ time for some computable functions $f$ and $f'$. Under standard complexity assumptions, it is assumed that no W[1]-hard problem is in FPT.}
\end{corollary}

This leaves open whether \textsc{Global Strong Stability} for symmetric approvals is in P for any $\alpha$ and whether it is in FPT or W[1]-hard with respect to~$\alpha$.
Reducing from \textsc{Independent Set}, we answer both questions negatively:

\begin{restatable}{proposition}{strongglobal}
	\label{th:strong-global}
	\textsc{Global Strong Stability} for symmetric bipartite approvals is NP-hard and~W[1]-hard wrt. $\alpha$, even if each agent approves at most two agents in each layer.
\end{restatable}
\begin{proof}
	We reduce from \textsc{Independent Set} which is W[1]-hard parameterized by the size of the independent set to be found~\cite{DBLP:journals/tcs/DowneyF95}.
	Assume we are given a graph $G=(V=\{v_1,\dots v_{\nu}\},E)$ and an integer $k$. 
	The set of agents $A$ consists of four agents~$e^1$, $e^2$, $e^3$, and $e^4$ for each edge~$e\in E$. 
	Moreover, for each $i\in [\nu]$ we introduce one layer $i$ (in which only agents corresponding to edges incident to vertex~$v_i$ approve each other).
	In layer~$i$, for each edge~$e=\{v_i,v_j\}\in E$ with~$i<j$, agents~$e^1$ and~$e^2$ as well as agents~$e^3$ and $e^4$ approve each other.
	Furthermore, for each edge~$e=\{v_i,v_j\}\in E$ with~${j<i}$, agents~$e^1$ and~$e^3$ as well as~$e^2$ and~$e^4$ approve each other in layer~$i$. 
	We set~$\alpha:=k$. 
	
	{\bfseries ($\Rightarrow$)} 
	Assume we are given an independent set $V'=\{v_{i_1}, \dots v_{i_k} \}\subseteq V$ in $G$.
	Let $M$ be the set of agent pairs~$\{a,a'\}$ such that $a$ and $a'$ approve each other in some layer from $\{i_1,\dots i_k\}$.
	Note that each agent $e^r\in A$ with $e=\{v_i,v_j\}\in E$ approves one agent in layer~$i$ and one agent in layer~$j$ and no agent in all other layers.
	Thus, as $V'$ is an independent set, for each agent~$e^r$ there is at most one layer from $\{i_1,\dots i_k\}$ where $e^r$ approves some (and in fact always exactly one) agent. 
	Thus, in $M$, each agent is only included in at most one pair and $M$ is a valid matching. 
	Moreover, $M$ is clearly strongly stable in all layers from $\{i_1,\dots i_k\}$ as for every agent~$a$ accepting another agent~$b$ in one of these layers, we have that~$\{a, b\} \in M$.
	
	{\bfseries ($\Leftarrow$)}
	Assume we are given a matching $M$ that is strongly stable in layers $\{i_1,\dots i_k\}$.
	We claim that $V'=\{v_{i_1}, \dots, v_{i_k}\}$ is an independent set in $G$. 
	Assume for the sake of contradiction that $e=\{v_{i_p},v_{i_q}\}\in E$ for $i_p<i_q$ with $p,q\in [k]$. 
	Note that $e^1$ cannot be happy in both layer~$i_p$ and layer~$i_q$, as it only approves $e^2$ in layer $i_p$ and only $e^3$ in layer $i_q$. 
	If $e^1$ is unhappy in layer~$i_p$ (respectively $i_q$), then it forms a blocking pair together with~$e^2$ (respectively $e^3$) in layer~$i_p$ (respectively~$i_q$), as $e^1$ and $e^2$ (respectively $e^3$) approve each other and no other agents in this layer, contradicting the stability of~$M$ in layer~$i_p$ (respectively~$i_q$). 
\end{proof}

Concerning results for arbitrary constellations of $\alpha$ and $\ell$ note that for every constant value of $\alpha$ or $\ell-\alpha$ the problem becomes polynomial-time solvable.

\paragraph{Pair Stability.}  \label{sub:strong-pair}
\newcommand{\xtrue}{x^{\text{true}}}
\newcommand{\xfalse}{x^{\text{false}}}
\newcommand{\xf}{x^{f}}
\newcommand{\xd}{x^{\text{d}}}

In contrast to \textsc{Pair Weak Stability}, \textsc{Pair Strong Stability} is NP-hard for any~${0 < \alpha < \ell}$ (not just~$\alpha > \lceil \ell / 2 \rceil$). 
\begin{restatable}{theorem}{pairstrong}
	\label{th:pair-strong}
	\textsc{Pair Strong Stability} for symmetric bipartite approvals is NP-hard for any~${\ell \geq 2}$ and any~$0 < \alpha < \ell$. 
\end{restatable}
\begin{proof}
	In order to show NP-hardness, we reduce from \textsc{3-Sat}.
	We always complete approvals to be symmetric.
	We will construct an instance in which agents only approve other agents in the first~$\ell - \alpha + 1$ layers.
	We call the remaining~$\alpha - 1$ last layers the \emph{empty layers}.
	Note that each matching is strongly stable in any empty layer.
	
	For each variable~$x$, we add four agents~$\xtrue$, $\xfalse$, $\xf$, and~$\xd$.
	In the first layer, agents~$\xtrue$ and $\xfalse$ approve~$\xd$.
	In the first~$\ell - \alpha + 1$ layers, agents~$\xtrue$ and $\xfalse$ approve~$\xf$ (see \Cref{fig:clause-gadget-pair-strong}).
	For each clause~$c$, we add six agents~$c^1$, $c^2$, $c^3$, $\hat c^1$, $\hat c^2$, and~$\hat c^*$.
	Agents~$\hat c^1$ and $\hat c^2$ approve~$c^1$, $c^2$, and $c^3$ in the first~$\ell - \alpha + 1$ layers, while $\hat c^*$ approves $c^1$, $c^2$, and~$c^3$ only in the second layer.
	Intuitively speaking, the partner of~$\xf$ indicates the truth value assignment to $x$ and the partner of~$\hat c^*$ indicates which literal is satisfied in $c$.
	If the $j$th literal of~$c$ is $\bar x$ for some variable~$x$, then~$c^j$~additionally approves $\xfalse$ in the first~$\ell - \alpha + 1$~layers.
	If the $j$th literal of~$c$ is $x$ for some variable~$x$, then $c^j$ additionally approves $\xtrue$ in the first~$\ell - \alpha + 1$~layers (see \Cref{fig:clause-gadget-pair-strong}).
	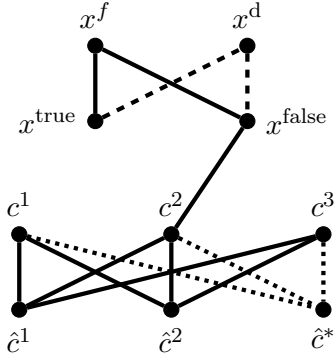
\begin{figure}[t]
		\begin{center}
			\begin{tikzpicture}
			\node[vertex, label=180:$\xtrue$] (xtrue) at (0,0) {};
			\node[vertex, label=0:$\xfalse$] (xfalse) at ($(xtrue) + (2,0)$) {};
			
			\node[vertex, label=90:$\xf$] (xf) at ($(xtrue) + (0, +1)$) {};
			\node[vertex, label=90:$\xd$] (xd) at ($(xf) + (2,0)$) {};
			
			\draw (xtrue) edge[ultra thick, dashed] (xd);
			\draw (xfalse) edge[ultra thick] (xf);
			
			\draw (xf) edge[ultra thick] (xtrue);
			\draw (xd) edge[ultra thick, dashed] (xfalse);
			
			\node[vertex, label=90:$c^1$] (c1) at ($(xtrue) + (-1, -1.5)$) {};
			\node[vertex, label=90:$c^2$] (c2) at ($(c1) + (2,0)$) {};
			\node[vertex, label=90:$c^3$] (c3) at ($(c2) + (2,0)$) {};
			
			\node[vertex, label=270:$\hat c^1$] (c1h) at ($(c1) + (0, -1)$) {};
			\node[vertex, label=270:$\hat c^2$] (c2h) at ($(c1h) + (2,0)$) {};
			\node[vertex, label=270:$\hat c^*$] (c3h) at ($(c2h) + (2,0)$) {};
			
			\draw (c1) edge[ultra thick] (c2h);
			\draw (c2) edge[ultra thick] (c1h);
			
			\draw (c1h) edge[ultra thick] (c1);
			\draw (c2h) edge[ultra thick] (c2);
			
			\draw (c3) edge[ultra thick] (c2h);
			\draw (c3) edge[ultra thick] (c1h);
			
			\draw (c1) edge[ultra thick, dotted] (c3h);
			\draw (c2) edge[ultra thick, dotted] (c3h);
			\draw (c3) edge[ultra thick, dotted] (c3h);
			
			\draw (c2) edge[ultra thick] (xfalse);
			\end{tikzpicture}
			
		\end{center}
		\caption{Approval graph of agents added for a variable~$x$ and a clause~$c$ and their approvals, where~$\bar x$ is the second literal of~$c$.
			Approvals appearing only in layer~1 are dashed, while approvals only appearing in layer~2 are dotted. All other approvals appear in the first~$\ell - \alpha + 1$ layers. 
		}
		\label{fig:clause-gadget-pair-strong}
	\end{figure}
	Finally, we add four agents~$a$, $a^*$, $b$, and $b^*$.
	Agent~$a$ approves $a^*$ and~$b$ approves $b^*$ in the first~$\ell - \alpha + 1$~layers.
	Agents~$a$ and $b$ approve some of the other agents as follows:
	For each variable~$x$, agent~$a$ approves agents $\xf$ and $\xd$ in the first~$\ell - \alpha + 1$ layers.
	For each clause~$c$, agent~$b$ approves~$\hat c^1$, $\hat c^2$, and~$\hat c^3$ in the first~$\ell - \alpha + 1$ layers.
	It is easy to verify that the approval graph is bipartite (with~$a$, $b^*$, $\xtrue$, $\xfalse$, $\hat c^1$, $\hat c^2$, and~$\hat c^3$ on one side of the bipartition and~\mbox{$a^*$, $b$, $\xf$, $\xd$, $c^1$, $c^2$}, and $c^3$ on the other side of the bipartition).
	
	$(\Rightarrow)$
	Assume that there is a satisfying assignment~$f$.
	We construct a $\alpha$-pair strongly stable matching~$M$ as follows.
	First, we add~$\{a, a^*\}$ and $\{ b, b^* \}$ to~$M$.
	For each variable~$x$ which is set to true by~$f$, we add $\{\xtrue, \xf\}$ and $\{\xfalse, \xd\}$ to~$M$.
	For each variable~$x$ which is set to false by~$f$, we add $\{\xfalse, \xf\}$ and~$\{\xtrue, \xd\}$ to~$M$.
	For each clause~$c$, fix $j_c \in [3]$ such that the $j_c$th literal of~$c$ is satisfied by~$f$.
	We add~$\{c^{j_c}, \hat c^*\}$ to~$M$.
	Moreover, we add two more pairs: $\{c^1, \hat c^1\}$ and $\{c^2, \hat c^2\}$ if $j_c = 3$, $\{c^1, \hat c^1\}$ and~$\{c^3, \hat c^2\}$ if $j_c = 2$, and~$\{c^2, \hat c^1\}$ and $\{c^3, \hat c^2\}$ if~${j_c = 1}$.
	This finishes the construction of~$M$.
	
	We now show that~$M$ is $\alpha$-pair strongly stable.
	Since every agent approves its partner in~$M$ in at least one of the first two layers, both $a$ and $b$ are happy in the first~$\ell - \alpha + 1$~layers, and no agent approves another agent in the last~$\alpha - 1$ layers, no pair containing~$a$ or $b$ is blocking in more than~$\ell - \alpha$ layers.
	For any variable~$x$, no pair consisting only of agents from~$\xtrue$, $\xfalse$,~$\xf$, and~$\xd$ is blocking in layer one, since all these agents are happy in layer one.
	For any clause~$c$, no pair consisting only of agents from~$c^1$, $c^2$, $c^3$, $\hat c^1$, $\hat c^2$, $\hat c^3$ is blocking in layer two, since all these agents are happy in layer two.
	Hence, no such pair of agents builds a blocking pair in~$\ell - \alpha + 1$ layers. 
	Consider a clause~$c$.
	Agents $\hat c^1$, $\hat c^2$ and their partners are happy in the first~$\ell - \alpha + 1$ layers and thus they are not part of a blocking pair.
	Assume without loss of generality that the~$j_c$th~literal of~$c$ is~$\bar x$ for some variable~$x$.
	Then, $c^{j_c}$ can only form a blocking pair with~$\xfalse$.
	By the choice of $j_c$, the variable $x$ is set to false by $f$.
	It follows that $M (\xfalse) = \xf$, which makes $\xfalse$ happy in the first~$\ell - \alpha + 1$ layers.
	Thus, $\{\xfalse, c^{j_c}\}$ is not blocking in layer~2 and thus it blocks in at most~$\ell - \alpha$ layers.
	Analogously, we can show that there is no blocking pair for the case that $j_c$th literal is positive.
	
	$(\Leftarrow)$
	Assume that there is an $\alpha$-pair strongly stable matching~$M$.
	First note that~$M$ needs to contain~$\{a, a^*\}$ and $\{ b, b^* \}$ since these pairs block $M$ otherwise in the first~$\ell - \alpha + 1 > \ell - \alpha$ layers.
	It follows that for each variable~$x$, matching~$M$ contains either~$\{\xtrue, \xf\}$ and $\{\xfalse, \xd\}$ or $\{\xfalse , \xf\}$ and~$\{\xtrue, \xd\}$ as otherwise~$\{\xf, a\}$ or $\{\xd, a\}$ would block~$M$ in the first~$\ell - \alpha + 1$ layers.
	We claim that setting variable~$x$ to true if and only if~$M$ contains~$\{\xtrue, \xf\}$ is a satisfying assignment.
	Consider a clause~$c$.
	Note that $M$ must match~$\hat c^1$, $\hat c^2$, and $\hat c^*$ to $c^1$, $c^2$, and $c^3$ as otherwise~$\{b, \hat c^1\}$, $\{b, \hat c^2\}$, or $\{b, \hat c^*\}$ would block~$M$ in the first~$\ell - \alpha + 1$ layers.
	Let $j_c \in [3]$ such that $M(\hat c^*) = c^{j_c}$.
	We claim the $j_c$th literal of~$c$ satisfies~$c$.
	Suppose that the $j_c$th literal of~$c$ is~$x$.
	Suppose further that $\{\xtrue, \xf\} \notin M$. Then, $c^{j_c}$ is unhappy in layer one and $\xtrue$ is unhappy in layers two to~$\ell - \alpha + 1$ and thus $\{c^{j_c}, \xtrue\}$ blocks~$M$ for $\alpha$-pair strong stability.
	It follows that $\{ \xtrue, \xf \} \in M$, implying that $c$ is satisfied.
	Otherwise the $j_c$th literal of~$c$ is $\bar x$ for some variable~$x$.
	Suppose that $\{\xfalse, \xf\} \notin M$, then $c^{j_c}$ is unhappy in layer one and $\xfalse$ is unhappy in layers two to~$\ell-\alpha + 1$ and thus $\{c^{j_c}, \xfalse\}$ blocks~$M$ for $\alpha$-pair strong stability.
	It follows that $\{ \xfalse, \xf \} \in M$, implying that $c$ is satisfied. 
\end{proof}
Notice that finding an all-layers stable matching is NP-hard for weak stability and symmetric approvals but finding an $1$-pair weakly stable matching is polynomial-time solvable.
The picture is reversed for strong stability and symmetric approvals.

\section{Super Stability}\label{se:super}
As for strong stability, in the single-layer setting, a super stable matching might not exist, but its existence can be decided in polynomial-time \cite{DBLP:journals/jal/IrvingM02}.
\paragraph{All-Layers and Global Stability.}
We show that \textsc{Global Super Stability} is polynomial-time solvable (even if approvals are asymmetric).
The key ingredient to this proof is the observation that each layer has at most three super stable matchings.
This suggests that achieving super stability is the algorithmically easiest among our three studied stability notions.

\begin{restatable}{proposition}{superglobal}
	\label{th:super-global}
	\textsc{Global Super Stability} can be solved in $O(\ell n^2)$ time.
\end{restatable}
\begin{proof}
	Let $G$ be a directed approval graph on $n$ agents (for the symmetric case, we add bidirectional arcs).
	We first describe a procedure to find all super stable matchings or determine that there is none in the layer corresponding to $G$.
	Observe that every bidirectional arc must be included in a super stable matching since the endpoints block otherwise.
	As a consequence, we immediately conclude that $G$ has no super stable matching if there is an agent incident to two bidirectional arcs.
	Hence, to find a super stable matching~$M$, as long as there is a bidirectional arc between say $u$ and $v$, we add $\{ u, v \}$ to $M$.
	Now let $G'$ be the result of deleting all bidirectional arcs  and incident agents.
	If there remain at least four agents in $G'$, we may conclude that there is no super stable matching in $G$:
	Consider an arbitrary maximal matching~$M'$ of agents~$V(G')$.
	Note that $M'$ has size at least two.
	Since~$G'$ has no bidirectional arcs, each pair in $M'$ leaves at least one agent unhappy.
	Hence, there are two unhappy agents who are not partners in $M'$.
	Note that these two agents always block even if there is no arc in between.
	Thus, we see that $G$ has no super stable matching if $G'$ has more than three agents.
	We thus assume that $G'$ has at most three agents.
	
	Recall that we obtained $G'$ by deleting bidirectional arcs and incident agents from~$G$.
	Consequently, a super stable matching in $G$ has a one-to-one correspondent super stable matching in $G'$ (since an agent in~$G'$ can never form a blocking pair with an agent not contained in~$G'$).
	Since $G'$ has at most three agents, there are at most three super stable matchings in any layer and all super stable matchings in~$G'$ can be found in constant time.
	Moreover, the construction of~$G'$ can be computed in $O(n^2)$ time.
	
	We next show how to solve \textsc{Global Super Stability}.
	First, we enumerate all super stable matchings for each layer.
	This takes $O(\ell n^2)$ time as we spend $O(n^2)$ for each layer.
	Assuming that $A = [n]$, we encode each matching as an $n$-digit number in base $n$, where the $i$-th digit is the partner of agent $i$. 
	We sort these numbers using radix sort in $O(\ell n + n^2)$ time.
	We conclude that the given instance is a yes-instance if and only if there is an element occurring at least~$\alpha$~times in this sorted sequence. 
\end{proof}

\paragraph{Pair and Individual Stability.}
In the following, we show a result similar to \cref{th:pair-weakP} which was concerned with weak stability.
However, the role of~$\alpha < \ell / 2$ and~$\alpha > \ell / 2$ are reversed for super compared to weak stability.
Intuitively, this comes from the fact that any matching is weakly stable if no agents approve other agents (or all agents approve each other) while no matching is super stable in these settings.
We first prove that both checking for a  pair or individually super stable matching is NP-hard for symmetric approvals for only two layers: 

\begin{restatable}{theorem}{superpairneg}
	\label{th:super-pair-neg}
	\textsc{Pair/Individual Super Stability} are NP-hard for symmetric approvals for any~$\ell \geq 2$ and any~${\alpha \leq \ell / 2}$, even if each agent approves at most four agents in each layer.
\end{restatable}
\begin{proof}
	We first show that \textsc{Pair Super Stability} is NP-hard for symmetric approvals for~$\ell = 2$ and~$\alpha = 1$ even if each agent approves at most four agents in each layer.
	Note that $1$-pair and~$1$-individually super stability are equivalent.
	We reduce from the following problem: Given a graph~$G=(V,E)$ with~$|V| $ even, decide whether there is a partition of~$V=V'\cupdot V''$ such that in $G[V']$ and in $G[V'']$ all vertices have degree one. Schaefer \cite{DBLP:conf/stoc/Schaefer78} in Theorem 7.1 proved that this problem is NP-hard even on cubic graphs.
	Note that $1$-pair super stability requires for each unmatched pair of agents that there is one layer where the two agents do not approve each other and at least one of the two is happy.
	
	From $G=(V,E)$, we construct an instance of \textsc{Pair Strong Stability} with two layers and $\alpha=1$ as follows.  
	For each vertex $v\in V$, we introduce three agents $v^1,v^2$ and~$v^*$. 
	In the first layer, agents~$v^1$ and $w^1$ approve each other and $v^2$ and $w^2$ approve each other for every edge~$\{v,w\}\in E$. 
	In the second layer, agents~$v^1$ and $v^*$ and $v^2$ and $v^*$ approve each other for each $v\in V$. 
	Moreover, we add two agents $a$ and $a'$ who, in both layers, do not approve any agents.
	
	{\bfseries ($\Rightarrow$)} 
	Assume we are given a partitioning $V=V'\cupdot V''$ such that in $G[V']$ and $G[V'']$ all vertices have degree one.
	For each vertex $v$, let $\beta(v)$ be $v$'s only neighbor in $v$'s partition, i.e., if $v\in V'$, then $\beta(v)$ is the unique element in $\{w\in V'\mid \{v,w\}\in E\}$ and if $v\in V''$, then $\beta(v)$ is the unique element in $\{w\in V''\mid \{v,w\}\in E\}$. 
	We now match $v^1$ to $\beta(v)^1$ and $v^2$ to $v^*$ if~$v\in V'$~and symmetrically $v^2$ to $\beta(v)^2$ and $v^1$ to~$v^*$ if $v\in V''$. 
	We also match $a$ and $a'$.
	Note that no vertex~$v^*$ can be part of a blocking pair: 
	It approves only two agents in the second layer and no agent in the first layer. 
	Further, agent~$v^*$ is matched to one of the agents it approves in the second layer, while the other agent it approves is happy in the first layer. 
	Moreover, $a$ and~$a'$ are not part of any blocking pair since all other agents are happy in one of the two layers.
	Finally, for each edge $\{v,w\}\in E$, it either holds that $v^1$ and $w^1$ are matched to each other or that one of the two is happy in the second layer, implying that $v^1$ and $w^1$ cannot be a blocking pair in the second layer (an analogous statement also holds for $v^2$ and $w^2$). 
	It follows that the constructed matching is $1$-pair super stable.
	
	{\bfseries ($\Leftarrow$)} 
	Let $M$ be a $1$-pair super stable matching.
	Note that we have an even number of agents (as $|V|$ is even).
	Since $a$ and~$a'$ do not approve any agent, these two are unhappy.
	If these two are not matched to each other in $M$, then they form a blocking pair.
	It follows that $\{a, a'\} \in M$.
	Moreover, all other agents must be happy in at least one layer, since otherwise such an agent and $a$ (or $a'$) form a blocking pair.
	Hence, for each $v\in V$, agent~$v^*$ needs to be matched either to $v^1$ or $v^2$. 
	Let $V'$ contain all $v\in V$ with $\{v^1, v^*\} \in M$ and let $V''=V\setminus V'$. 
	Note that as each agent is happy in one layer it needs to hold for each $v\in V'$ that $v^1$ is matched to an agent~$w^1$ with $w\in V'$ such that $\{v,w\}\in E$, implying that each vertex in $G[V']$ has degree at least one. 
	Assume now for the sake of contradiction that there is some $u\neq w\in V'$ with $\{v,u\}\in E$, then~$v^1$ and~$u^1$ are both happy in the first layer and unhappy in the second layer, implying that they form a blocking pair. 
	It follows that all vertices in $G[V']$ have degree one and the same follows for~$G[V'']$ analogously.
	
	We next generalize this result to NP-hardness for \textsc{Pair/Individual Super Stability} for any~$\ell \geq 2$ and any~$\alpha \leq \ell / 2$.
	To this end, we copy both layers from the previous reduction~$\alpha$ times and add~$\ell - 2 \alpha$ \emph{empty} layers, that is, layers in which no agent approves another agent.
	Note that any pair of unmatched agents is a blocking pair in any empty layer.
	Hence, in order for a matching to be~$\alpha$-pair or~$\alpha$-individual super stable, it needs to be~$1$-pair or~$1$-individual super stable in the original construction.
	On the other hand, if a matching is~$1$-pair ($1$-individual) super stable, then it is~$\alpha$-pair and~$\alpha$-individually stable in the modified construction.
	Thus, the modified reduction is a yes-instance if and only if the original instance is a yes-instance.
	Observe that the approvals are still symmetric in the adjusted reduction and each agent still approves at most four agents in each layer.
	\end{proof}

Again, for certain combinations of $\alpha$ and $\ell$, tractability can be regained. 
\begin{restatable}{theorem}{superpairpos}
	\label{th:super-pair-pos}
	For~$\alpha > \ell / 2$ and symmetric approvals, \textsc{Individual Super Stability} is polynomial-time solvable and \textsc{Pair Super Stability} is~FPT parameterized by~$\ell$.
	\textsc{Pair Super Stability} with symmetric approvals is polynomial-time solvable if~$\alpha > 2\ell / 3$.
\end{restatable}
\begin{proof}
	We first show that an $\alpha$-individual super stable matching can be found in polynomial-time if~$\alpha > \ell / 2$ and approvals are symmetric.
	This algorithm will also be the basis for our results regarding~$\alpha$-pair super stable matchings.
	Consider the graph~\Gla{} with vertex set~$A$ and an edge~$\{u, v\}$ if and only if there are at least~$\ell - \alpha + 1$ layers in which~$u$ and~$v$ approve each other.
	For each edge~$\{u,v\}$ in~\Gla{}, the two agents~$u$ and~$v$ must be matched to each other in any~\mbox{$\alpha$-individual} super stable matching as they otherwise form a blocking pair in at least~$\ell - \alpha + 1$~layers ($u$ and~$v$ find each other at least as good as their matched partners in~$\alpha+1$~layers).
	Hence, if a vertex has degree at least two in~\Gla{}, then there is no solution.
	Moreover, we can match each agent with degree exactly one in~\Gla{} to its unique neighbor, resulting in a matching~$M'$. 
	If there is at most one unmatched agent in~$M'$, then we set~$M := M'$.
	If there are exactly two unmatched agents~$u$ and~$v$, then we set~$M := M' \cup \{\{u, v\}\}$.
	Otherwise, we set~$M:= \emptyset$.
	We claim that if there exists a super stable matching, then $M$ is one.
    In the first two cases, we already observed that $M$ is the only possible super stable matching.
	In the third case, i.e., there are at least three vertices of degree zero in~\Gla, we show there is no super stable matching.
	Assume towards a contradiction that there are at least three agents of degree zero in \Gla{} but there is an~$\alpha$-individual super stable matching~$M^*$.
	Then, there are at least two agents~$u$ and~$v$ that have degree zero in \Gla{} who are not matched to each other in~$M^*$.
	Since~$u$ has degree zero in \Gla{}, it approves $M^* (u)$ in at most~$\ell - \alpha < \ell / 2 < \alpha$ layers.
	Thus,~$u$ finds~$v$ at least as good as~$M^* (u )$ in at least~$\alpha + 1 > \ell - \alpha$ layers.
	Analogously, $v$ finds $u$ at least as good as~$M^* (v)$ in at least~$\alpha + 1 > \ell - \alpha$ layers.
	This contradicts the assumption that~$M^*$~is an~$\alpha$-individual super stable matching.
	Thus, \Gla{} contains at most two agents of degree zero and we can check whether there is an~$\alpha$-individual super stable matching in polynomial time.
	
	We now move to~$\alpha$-pair super stable matchings with $\alpha > 2\ell / 3$.
	We again consider the graph~\Gla{} (the graph containing an edge~$\{u,v\}$ if and only if~$u$ and~$v$ approve each other in at least~$\ell - \alpha + 1$ layers).
	Again, for each edge~$\{u,v\}$ in~\Gla{}, the two agents~$u$ and~$v$ must be matched to each other in any~$\alpha$-pair super stable matching as they otherwise form a blocking pair in~$\ell - \alpha + 1$ layers.
	We may further assume that no vertex has degree at least two in~\Gla.
	We next show that an~$\alpha$-pair super stable matching can be found in polynomial time if~$\alpha > 2 \ell / 3$.
	The argument is similar to the one above, that is, we show by contradiction that if there are at least three agents of degree zero in~\Gla, then there is no solution.
	Let~$M$ be the assumed~$\alpha$-pair super stable matching and let~$u$ and~$v$ be two agents of degree zero in~\Gla{} who are not matched to one another by~$M$.
	Let further~$u'$ be the partner of~$u$ in~$M$ and let~$v'$ be the partner of~$v$ in~$M$.
	Since~$u$ and~$v$ have degree zero in~$G$, they are each happy in at most~$\ell - \alpha < \ell / 3$ layers.
	Hence, there are at least~$\ell / 3 > \ell - \alpha$ layers where they are both unhappy, that is, they both find each other at least as good as their partner in~$M$.
	Thus, they form a blocking pair in at least~$\ell - \alpha + 1$ layers.
	This is a contradiction to the assumption that here is no blocking pair for~$M$ in at least~$\alpha$ layers.
	
	Finally, we show that~\textsc{$\alpha$-Pair Super Stability} is FPT parameterized by~$\ell$ if~$\alpha > \ell / 2$.
	To this end, we again consider the graph \Gla{} and show that if there are more than~$2^{\ell + 1}$ vertices of degree zero in \Gla, then there cannot be a solution.
	Assume towards a contradiction that this is not the case, that is, \Gla\ contains more than $2^{\ell +1} $ vertices and there is an~$\alpha$-pair super stable matching~$M$.
	Then, by the pigeonhole principle, there are three agents who are happy in the matching~$M$ in the same set of layers.
	Moreover, since they all have degree zero in~\Gla, they are all happy in at most~$\ell - \alpha$ layers.
	However, since at least two of the three agents are not matched to one another in~$M$ and since they are unhappy in the same set of at least~$\alpha \geq \ell - \alpha + 1$ layers, they form a blocking pair in~$\ell - \alpha +1$ layers.
	This contradicts the assumption that~$M$ is~$\alpha$-pair super stable. 
\end{proof}

Note that the FPT result for \textsc{Pair Super Stability} with~$\alpha > 2\ell / 3$ and symmetric approvals excludes NP-hardness for constant~$\ell$ (unless~$P = NP$).
We leave it open whether the positive results can be extended to asymmetric approvals.

\section{Similarity Leads to Tractability} \label{se:similarity}
To circumvent the NP-hardness results from the previous sections, we start an investigation into the parameterized complexity of our problems.
We focus on types of ``similarity'' in the agents' preferences and show for three different types that our problems become tractable when preferences are similar. 
Similar preferences might, for instance, occur in cases where layers correspond to objective criteria---the preferences of all agents might then be similar or even identical within a layer. We study this as uniform preferences.
This type of similarity has also already been extensively studied for different stable matching problems in the context of master lists (see e.\,g.\ \cite{DBLP:journals/dam/IrvingMS08,DBLP:conf/wine/BredereckHKN20}).

\paragraph*{Few Agent Types.}
We say that two agents $a$ and $a'$ are of the same agent type if in each layer $i\in [\ell]$, $a$ and $a'$ approve the ``same'' set of agents, i.e.,~${T^i_a \setminus \{a'\}=T^i_{a'} \setminus \{a\}}$ and ${a'\in T^i_{a}}$ if and only if $a\in T^i_{a'}$, and are approved by the same agents, i.e., $a\in T^i_b$ if and only if~$a'\in T^i_b$ for each $b\in A\setminus \{a,a'\}$ (note that the second condition is redundant in the symmetric setting). 
The number of agent types has proven to be a useful parameter for various stable matching problems~\cite{DBLP:journals/corr/abs-2112-05777,DBLP:journals/tcs/MeeksR20}, which is again the case here:
\begin{restatable}{theorem}{types}
	\label{pr:types}
	Let~$\tau$ be the number of agent types.
	\textsc{Global/Pair/Individual} \textsc{Weak/Strong/Super Stability} is solvable in $\mathcal{O}(2^{(\tau+1)^2}\cdot n^4\cdot \ell)$~time.
\end{restatable}
\begin{proof}
	The general approach of this proof is similar to Proposition 5 of Boehmer et al. \cite{DBLP:journals/corr/abs-2112-05777}.
	Let~$\mathcal{I}$~be the given instance.
	If the number of agents is odd, then we modify $\mathcal{I}$ by inserting a new agent $d$ that does not approve any agent and is not approved by any agent in each layer. 
	Note that doing so does not impact the existence of a stable matching: 
	For all our stability notions, if a currently unmatched agent~$a$ forms a blocking pair with an agent $b$ in some layer then also after matching $a$ to $d$, pair~$\{a,b\}$ is still blocking in this layer. 
	Moreover, a matching that is stable in some layer can be extended to a matching in the modified instance that is stable in this layer by matching an unmatched agent $a$ to $d$. 
	Assume, for the sake of contradiction, that there is a blocking pair $\{d,b\}$ in the resulting matching for some $b\in A$, then also the pair~$\{a,b\}$~is blocking in the original matching.
	Thus, we can assume without loss of generality that the number of agents is even and thus can restrict our attention to finding a perfect matching. 
	
	Let $\tau$ be the number of agent types in the original instance and
	let $T$ be the set of agent types in the modified instance (note that $\tau+1=|T|$) and for $t\in T$, let $A_t\subseteq A$ denote the set of agents of type $t$.
	
	We iterate over all undirected graphs~$G$ with self loops on~$T$  where each vertex is incident to at least one edge (there are $\mathcal{O}(2^{(\tau+1)^2})$ such graphs).
	We say that a matching $M$ is \emph{compatible} with $G$ if $\{a,a'\}\in M$ with $a\in A_t$ and $a'\in A_{t'}$ for some $t,t'\in T$ only if $\{t,t'\}\in E(G)$. 
	We reject the graph~$G$ if a matching which is compatible with $G$ can be unstable. 
	To check this, we create a new instance $\mathcal{J}$ of the considered problem and matching $N$ as follows.
	For each edge $\{t,t'\}\in E(G)$, we create an agent $a_{t,t'}$ and an agent $a_{t',t}$ and match them to each other in $N$.
	Concerning the agent's preferences, an agent~$a_{t,t'}$ for $t,t'\in T$ approves all agents $a_{t'',t'''}$ for $t'',t'''\in T$ such that agents of type $t$ approve agents of type $t''$ in $\mathcal{I}$. 
	Subsequently, we check whether $N$ fulfills the desired multilayer stability criterion in $\mathcal{J}$. 
	If this is not the case, then we continue with the next graph $G$; 
	otherwise, we create a new graph $G^*$ on~$A$ where we connect two agents $a,a'\in A$ with $a\in A_t$ and~$a'\in A_{t'}$ for some $t,t'\in T$ if $\{t,t'\}\in E(G)$. 
	We then check whether there is a perfect matching $M^*$ in $G^*$; 
	if this is the case we return yes and otherwise we continue with the next graph $G$. 
	
	Assume that the algorithm returns yes, then we claim that the computed perfect matching~$M^*$ in $G^*$ for graph $G$ is a stable matching in $\mathcal{I}$. 
	Assume for the sake of contradiction that there is a blocking pair $\{a,\hat{a}\}$ for $M^*$ in $\mathcal{I}$ under the relevant multilayer stability criterion with~$a\in A_t$ and $\hat{a}\in A_{\hat{t}}$ and that $a$ is matched to an agent $a'$ with~$a'\in A_{t'}$ and $\hat{a}$ is matched to an agent $\hat{a}'$ with $\hat{a}'\in A_{\hat{t}'}$ in~$M^*$. 
	Then, agents $a_{t,t'}$ and $a_{\hat{t},\hat{t}'}$ form a blocking pair for the matching $N$ in $\mathcal{J}$ constructed in the iteration where $M^*$ was constructed, as
	\begin{itemize}
		\item $a_{t,t'}$ approves $a_{\hat{t},\hat{t}'}$ if and only if $a$ approves $\hat{a}$,
		\item $a_{\hat{t},\hat{t}'}$ approves $a_{t,t'}$ if and only if $\hat{a}$ approves $a$,
		\item $a_{t,t'}$ approves~$N(a_{t,t'})= a_{t',t}$ if and only if $a$ approves~$M^*(a)=a'$, and
		\item $a_{\hat{t},\hat{t}'}$ approves $N(a_{\hat{t},\hat{t}'})= a_{\hat{t}',\hat{t}}$ if and only if $\hat{a}$ approves~$M^*(\hat{a})=\hat{a}'$.
	\end{itemize}
	This is a contradiction to a matching being returned for graph~$G$ as the algorithm rejected~$G$.
	
	Assume that there is a stable matching $M$ in $\mathcal{I}$.
	Without loss of generality we can assume that $M$ is perfect.
	Let $G$ be the graph where two types $t,t'\in T$ are connected if and only if there is an agent of type $t$ matched to an agent of type~$t'$ in $M$.
	We claim that $G$ was not rejected by the algorithm. 
	Assume for the sake of contradiction that $G$ was rejected because agents $a_{t,t'}$ and $a_{\hat{t},\hat{t}'}$ form a blocking pair. 
	Then, for such agents to exist there needs to be an agent $b$ of type $t$ that is matched to an agent of type $t'$ in $M$ and there needs to be an agent $c$ of type $\hat{t}$ that is matched to an agent of type $\hat{t}'$ in~$M$. 
	As
	\begin{itemize}
		\item $b$ approves $c$ if and only if $a_{t,t'}$ approves $a_{\hat{t},\hat{t}'}$,
		\item $b$ approves $M(b)$ if and only if $a_{t,t'}$  approves $N(a_{t,t'})$,
		\item $c$ approves $M(c)$ if and only if $a_{\hat{t},\hat{t}'}$  approves $N(a_{\hat{t},\hat{t}'})$, and
		\item $c$ approves $b$ if and only if $a_{\hat t, \hat t'}$ approves $a_{t, t'}$,
	\end{itemize}
	it follows that $b$ and $c$ block $M$ in $\mathcal{I}$ under the respective stability criterion, a contradiction. 
	Further, $M$ is clearly a perfect matching in $G^*$ and, thus, the algorithm returns yes.
	
	As there are $\mathcal{O}(2^{(\tau+1)^2})$ graphs $G$ over which we iterate and constructing the instance $\mathcal{J}$ and matching $N$ and checking whether $G^*$ admits a perfect matching can be done in $\mathcal{O}(n^4\cdot \ell)$ for each graph $G$, the overall running time of~$\mathcal{O}(2^{(\tau+1)^2}\cdot n^4\cdot \ell)$ follows. 
\end{proof}

\paragraph{Uniform Approvals.} 
Chen et al. \cite{DBLP:conf/sigecom/ChenNS18} showed polynomial-time solvability for some problems if agents' preferences are uniform, i.e., when within a layer all agents have the same preferences.
Hence, each agent is either approved by all or by no other agents in a layer.
For symmetric approvals, the situation becomes simple:
In each layer, either every pair of agents approves each other or every pair disapproves each other. 
Thus, for uniform symmetric approvals, all our problems are in P. 
For asymmetric approvals, we show that all our problems are in FPT with respect to $\ell$ using \Cref{pr:types} (it is open which of our problems become polynomial-time solvable for uniform asymmetric approvals).
\begin{corollary}
	\textsc{Global/Pair/Individual Weak/Strong/Super Stability} is in FPT wrt.~$\ell$ if in each layer each agent is either approved by all other agents or by no other agent.
\end{corollary}
\begin{proof}
	For $i\in [\ell]$ and $a\in A$, let $s^i_a$ be one if agent $a$ is approved by all agents in layer $i$. 
	Here an agent $a\in A$ is fully characterized by $(s^i_a)_{i\in[\ell]}$, as, within each layer, all agents approve the same agents. 
	Thus, there are only $2^{\ell}$ different agent types.
	The statement now follows from \Cref{pr:types}.
\end{proof}

\paragraph{Few Agents with Changing Preferences.}

Lastly, we turn to situations with only few ``changing'' agents, i.e., agents that do not approve the same set of agents in each layer. 
We focus on symmetric approvals.
The crucial observation here is that non-changing agents cannot be involved in a blocking pair in any layer because this pair would then block all layers:
\begin{restatable}{theorem}{changingprefs}
	\label{th:changing-prefs}
	Let $\beta$ be the number of agents whose approval sets are not identical in all layers. 
	For symmetric approvals, \textsc{Global/Pair/Individual Weak/Strong/Super Stability} is in FPT wrt.~$\beta$. 
\end{restatable}
\begin{proof}
	Let $B\subseteq A$ be the set of agents whose approval sets are not identical in all layers. 
	For~$b \in B$, let $C_b \subseteq A\setminus B$ be the set of agents from~$A\setminus B$ which $b$ approves (in all layers).
	
	We now make a case distinction for the three notions of stability. 
	
	\paragraph{Weak Stability.}
	We create a set $H$ of agents that need to be happy in all layers and a matching $M$ as follows.
	We start with $H:=\emptyset$ and $M:=\emptyset$. 
	For each agent $b\in B$, we guess whether $b$ is matched to an agent from $B$.
	If yes, then we guess to which agent from $B$ agent~$b$ is matched and add the pair to $M$. 
	If no, then we guess whether $b$ needs to be happy in all layers and if yes, add it to $H$. 
	Moreover, in both cases, we guess whether all agents from $C_b$ need to be happy in all layers or not and if yes, add $C_b$ to $H$. 
	If any of our guesses are in conflict with each other (e.\,g., we guess that an agent~$b\in B$ is matched to an agent from $A\setminus B$, yet guessed for some agent~$b'\in B\setminus \{b\}$ that it is matched to $b$), then we reject the current guess. 
	Let~$B'\subseteq B$~be the set of agents currently matched in~$M$.
	
	We create a graph $G$ containing the agents from~$A \setminus B'$ as vertices where we connect agent~${a\in A\setminus B}$ and $a'\in A\setminus B'$ if $a$ and $a'$ approve each other (note that whether~$a$ and~$a'$ approve each other is independent of the layer as $a \notin B$ and approvals are symmetric, and that we do not add edges between two agents from $B\setminus B'$ because we have guessed that they will not be matched to each other).
	We check whether there is a matching $N$ in $G$ that matches all agents from $H$ (by computing a maximum-weight matching). 
	If no such matching exists, then we reject the current guess. 
	If such a matching exists, then we extend~$N$ arbitrarily to a maximal matching in $G$ and set ${M^*:=M\cup N}$. 
	Finally, we return yes if $M^*$ fulfills the required stability criterion and otherwise reject the current guess. 
	
	It remains to prove that if there is a matching $M^*$ that fulfills the required stability criterion, then the algorithm returns yes. 
	We claim that the algorithm returns yes for the following guess starting with $\tilde{M}:=\emptyset$ and $\tilde{H}:=\emptyset$: 
	For each~$b\in B$ with~$M^*(b)\in B$, we guess that $b$ is matched to~$M^*(b)$ and add this pair to~$\tilde{M}$. 
	For each $b\in B$ with $M^*(b)\notin B$, we guess that $b$ is not matched to an agent from $B$ and guess that $b$ is part of $\tilde{H}$ if $b$ approves $M^*(b)$ in some layer (as~$b$ is matched to an agent from $A\setminus B$ and approvals are symmetric, this implies that $b$ approves~$M^*(b)$ in all layers). 
	For each~$b\in B$, we add $C_b$ to $\tilde{H}$ if for each $a\in C_b$, agent~$a$ approves $M^*(a)$ in some layer (as $C_b\subseteq A\setminus B$ this implies that $a$ approves $M^*(a)$ in all layers). 
	
	Let $G$ be the graph constructed based on $\tilde{M}$ and $\tilde{H}$. 
	Note that the set of agents~${B'\subseteq B}$ already matched by~$\tilde{M}$ is exactly the set of agents from~$B$ that are matched to agents from~$B$ in~$M^*$. 
	Moreover, note that all agents from $\tilde{H}$ are matched to an agent they approve in all layers in $M^*$ and are in particular matched to an agent from $A\setminus B'$ in $M^*$.
	Thus, the matching~$M^*$ restricted to $G$ induces a matching that matches all agents from $\tilde{H}$. 
	
	Let $N$ be some matching that matches all agents from $\tilde{H}$ in $G$ (as argued above such a matching is guaranteed to exist) and that is maximal in $G$.
	We add $N$ to $\tilde{M}$. 
	We now claim that $\tilde{M}$ fulfills the desired stability criterion. 
	For the sake of contradiction assume that~$\tilde{M}$~admits a blocking pair $\{a,a'\}$ for the considered multilayer stability notion. 
	We make a case distinction.
	
	If $a,a'\in A\setminus B$, then $a$ and $a'$ can only be blocking if they approve each other in one (and thereby all layers) and if they are either unmatched or matched to an agent they do not approve in $\tilde{M}$. 
	As all agents from $A\setminus B$ are either unmatched or matched to an agent they approve in all layers in $\tilde{M}$, this implies that both $a$ and $a'$ are unmatched in $\tilde{M}$ and thus that $\tilde{M}$ is not maximal in $G$, a contradiction. 
	
	Before considering the case that $a, a' \in B$, we observe that whenever an agent $b\in B$ is happy in some layer~$i$ in~$M^*$, then $b$ is also happy in layer~$i$ in $\tilde{M}$:
	If $b\in B'$, then ${M^*(b)=\tilde{M}(b)}$.
	Otherwise, if $b$ is happy in some layer in~$M^*$, then $b\in \tilde{H}$ and thus $b$ is happy in all layers in $\tilde{M}$. 
	
	If $a,a'\in B$, then they also form a blocking pair for $M^*$, since if they are unhappy in a layer in~$\tilde{M}$, then they are also unhappy in~$M^*$. 
	
	It remains to consider the case $a\in A\setminus B$ and $a'\in B$. 
	First note that $a$ and $a'$ need to approve each other in at least one (and in fact all) layers to be able to form a blocking pair for~$\tilde{M}$.
	As already argued above, $a'$ is happy in $\tilde{M}$ in each layer in which it is happy in $M^*$. 
	Thus for $\{a,a'\}$ not to block~$M^*$, there needs to be a layer where $a$ is happy in $M^*$ but not in~$\tilde{M}$. 
	This implies that $a$ cannot be happy in all layers in $\tilde{M}$ and thus cannot be part of~$\tilde{H}$. 
	As~$a\in C_{a'}$, this implies that there is some $a^*\in C_{a'}$ that is not happy in all (and, as $a^*\in A \setminus B$, not happy in any) layer in $M^*$. 
	However, as $a'$ approves both $a$ and $a^*$ in all layers, this implies that $\{a^*, a'\}$ blocks $M^*$ under the considered multilayer stability notion, a contradiction. 
	
	\paragraph{Strong and Super Stability.}
	We create a matching $M$ as follows.
	We start with $M:= \emptyset$. 
	For each agent $b\in B$, we guess whether $b$ is matched to an agent from $B$.
	If yes, then we guess to which agent from $B$ agent $b$ is matched and add the pair to $M$. 
	We reject the current guess if it includes a conflict. 
	Let $B'\subseteq B$ be the set of agents currently matched by $M$.
	
	We create a graph $G$ containing the agents from~$A \setminus B'$ as vertices.
	Moreover, we connect agents $a\in A\setminus B$ and $a'\in A\setminus B'$ if $a$ and $a'$ approve each other in some (and thereby all) layers. 
	We compute a maximum-cardinality matching $M'$ in $G$ and add to $M'$ an arbitrary mapping of all agents from $A\setminus B'$ that are currently unmatched in $M'$ (potentially leaving one agent unmatched).
	We return yes if $M'\cup M$ fulfills the desired stability criterion and reject the current guess otherwise. 
	
	It remains to show that if there is a matching $M^*$ that fulfills the desired stability criterion, then the algorithm returns yes.
	We claim that this is the case for the following guess: 
	We guess for each $b\in B$ with $M^*(b)\in B$ that $b$ is matched to $M^*(b)$ and add this pair to $M$ and for all $b\in B$ with $M^*(b)\notin B$ that $b$ is not matched to an agent from $B$. 
	Let $G$ be the graph constructed from this guess and $M'$ a maximum-cardinality matching in $G$ extended by an arbitrary matching of so-far unmatched agents. 
	Set $\hat M:=M\cup M'$. 
	We claim that $\hat M$ fulfills the desired stability criterion. 
	
	\paragraph{Strong Stability.}
	Let $x$ be the number of agents from $A\setminus B'$ that are incident to at least one edge in $G$.
	We claim that $|M^*\cap E(G)| = \frac{x}{2}$ (i.e., $M^* \cap E(G)$ is a perfect matching of the set of agents from~$A \setminus B'$ which are incident to at least one edge in~$G$).
	If this is not true, then there is either an agent $a\in A\setminus B$ that is not matched to an agent it approves in~$M^*$  and which approves an agent $a'\in A\setminus B'$ or an agent~$a\in B\setminus B'$ that is not matched to an agent it approves in $M^*$  and which approves an agent $a'\in A\setminus B$.
	In both cases, $\{a,a'\}$ form a blocking pair in all layers. 
	Thus, $M^*$ cannot fulfill the desired stability criterion, a contradiction.
	Because~$M'$ is a maximum-cardinality matching in~$G$, it follows that $|M' | = \frac{x}{2}$.
	In other words, $M'$ matches all vertices in $G$ that are incident to at least one edge.
	Thus, all agents from $A\setminus B$ that approve at least one agent from $A\setminus B'$ and all agents from $B\setminus B'$ that approve at least one agent from~$A\setminus B$ are happy in all layers in $\hat M$.
	Thus, any blocking pair for~$\hat M$ needs to either involve an agent from $B'$ or be between two agents from $B\setminus B'$.
	We make a case distinction.
	
	First, assume that $a,a'\in B'$ form a blocking pair for the considered multilayer stability notion. 
	However, as $\hat M(a)=M^*(a)$ and $\hat M(a')=M^*(a')$ from this it follows that $\{a,a'\}$ also blocks $M^*$, contradicting the stability of~$M^*$.
	
	Second, assume that~$a\in B'$ and~$a'\in A\setminus B'$.
	If~$a'$~is happy in all layers in~$\hat M$, then~$a$ and~$a'$~also block~$M^*$ as~$M^*(a')=M(a')$. 
	Otherwise, by our above observation it follows that~$a'$~does not approve any agents from~$A\setminus B'$ and as all agents from $B'$ are matched the same in~$M$ and~$M^*$, it follows that $a'$ is unhappy in all layers in $M^*$. 
	As $\hat M(a)=M^*(a)$ it follows that $\{a,a'\}$ also blocks $M^*$. 

	Third, assume that $a,a'\in B\setminus B'$.
	By the construction of our guess, both $a$ and $a'$ are not matched to an agent from $B$ in~$M^*$.
	If $a$ or $a'$ is unhappy in some (and thereby all) layers in~$M$, then as we have observed above $a$ and $a'$ can only be unhappy in some layer if they do not approve any agents from $A\setminus B$. 
	Thus, as both are not matched to agents from $B$ in $M^*$, this agent is also unhappy in all layers in $M^*$.
	Consequently, if~$a$ and~$a'$ are a blocking pair in a layer in~$\hat M$, then they are also a blocking pair in the same layer in~$M^*$.
	
	\paragraph{Super Stability.}
	Note that for $M^*$ to be stable, all agents from $A\setminus B$ can approve at most one agent in any (and by construction all) layers (as in case they approve two agents, they form a blocking pair together with the one they are not matched to in all layers).
	Moreover, each agent from~$A\setminus B$ that approves at least one agent needs to be matched to it, as otherwise they form a blocking pair in all layers. 
	Thus, $M^*$ and $\hat M$ both contain all edges from $G$.
	Moreover, for each~$b\in B'$, it holds that $M^*(b)= \hat M(b)$.
	Note further that in $G$ only at most two agents can have no neighbor: 
	If there are three such agents, then as for each~$b\in B\setminus B'$ it holds that~$M^*(b)\notin B$, there are three agents in~$M^*$ that are unhappy in all layers and thus at least one pair of these three agents is not matched to each other and forms a blocking pair in all layers. 
	Further note that if there are two agents in~$G$ that have no neighbor, then $M^*$ matches them together, as they otherwise form a blocking pair in all layers. 
	As $\hat M$ in the end matches so-far unmatched agents together, $\hat M$ and $M^*$ are thus identical, implying that $\hat M$ is stable. 
\end{proof}

For asymmetric approvals, obtaining even an XP-algorithm is not possible for weak stability.%
\begin{restatable}{proposition}{changingprefshard}
	\label{th:changing-prefs-hard}
	\textsc{All-Layers Weak Stability} is NP-hard for any~$\ell \geq 2$ and \textsc{Individual Weak Stability} is NP-hard for $\ell=\alpha=2$.
	Both results hold for bipartite approvals, even if there is only one agent whose approval set differs between the two layers. 
\end{restatable}
\begin{proof}
	We start by focusing on all-layers weak stability and reduce from the NP-hard \textsc{Minimum Maximal Matching} problem on bipartite graphs, where given a bipartite graph $G$ and an integer $k$, the question is whether there is a maximal matching containing at most~$k$~edges \cite{yannakakis1980edge}.
	
	Given an instance $\mathcal{I}=(G=(V\cupdot U,E),k)$ with $|U|=|V|=n_V$ of \textsc{Minimum Maximal Matching}, we construct an instance $\mathcal{J}$ of \textsc{All-Layers Weak Stability} as follows. 
	We assume without loss of generality that a maximum matching in~$G$ contains at least~$k$ edges.
	For each vertex~${v\in V\cup U}$, we introduce a vertex agent $a_v$. 
	Moreover, we add~${n_V-k}$~penalizing agents~$p_1,\dots, p_{n_V-k}$.
	Lastly, we add two special agents $b$ and $c$. 
	In all layers, agents $a_v$ and~$a_{v'}$ for $v,v'\in V$ approve each other if $\{v,v'\}\in E$. 
	Each penalizing agent approves all agents from~$\{a_v\mid v\in V \}$ and~$b$ in all layers. 
	Agent~$c$ approves agent $b$ in all layers. 
	Agent $b$ approves agent $c$ in the first layer and all penalizing agents in all other layers. 
	Note that all agents except agent $b$ have the same preferences in all layers. 
	
	$(\Rightarrow)$ Given a maximal matching $M$ containing at most $k$ edges in $G$, we construct an all-layers weakly stable matching~$N$ in $\mathcal{J}$ as follows. 
	First, we assume without loss of generality that~$M$ contains exactly~$k$ edges (as maximal matchings are interpolating, i.e., whenever there exists maximal matchings of sizes~$i$ and~$j$ ($i < j$), respectively, then there exists a maximal matching of size~$i^*$ for all $i< i^* <j$).
	Let $v_1,\dots, v_{n_V-k}$ be the $n_V-k$ vertices from $V$ that are unmatched in $M$. 
	Let $N:=M\cup\{\{v_i,p_i\}\mid i\in [n_V-k]\}\cup \{\{b,c\}\}$. 
	Note that all agents except~$b$~and agents corresponding to vertices that are unmatched by~$M$ are happy in all layers in $N$. 
	As $b$ does not approve any vertex agent, it follows that a blocking pair needs to consist of two vertex agents corresponding to vertices unmatched by $M$.
	For two such agents $a_v$ and $a_{v'}$ to be blocking, they need to approve each other in some (and by construction also all) layers.
	Thus,~$\{v,v'\}\in E$~needs to hold and $M$ leaves both $v$ and $v'$ unmatched. 
	However, this contradicts the maximality of $M$ in $G$ and thus no such pair of vertex agents can exist.
	
	$(\Leftarrow)$ 
	Given an all-layers weakly stable matching $N$ in $\mathcal{J}$, we construct a maximal matching~$M$ containing at most $k$ edges in $G$ as follows. 
	First, observe that in $N$, agents~$b$ and~$c$ need to be matched to each other, as they only approve each other in the first layer and thus would form a blocking pair in this layer otherwise. 
	From this it follows that all penalizing agents need to be happy in all layers but the first and thus matched to a vertex agent from~${\{a_v\mid v\in V \}}$, as otherwise they form a blocking pair together with $b$ in the all layers but the first. 
	Let~${M:=\{\{v,v'\}\mid \{a_v,a_{v'}\}\in N\}}$. 
	
	First, matching $M$ contains at most $k$ edges: As $|\{a_v\mid v\in V \}|=n_V$  and as observed above all but $k$ of them are matched to a penalizing agent, only $k$ agents from $\{a_v\mid v\in V \}$ can be matched to another vertex agent. 
	As $G$ is bipartite, $M$ can contain at most $k$ edges.
	
	Second, matching $M$ is maximal: For the sake of contradiction assume that this is not the case because for some~$\{v,v'\}\in E$ both $v$ and $v'$ are unmatched in~$M$. 
	This implies that both~$a_v$~and~$a_{v'}$ are not matched to a vertex agent in $N$, which means that they are unhappy in all layers in~$N$. 
	However, as $\{v,v'\}\in E$, agents~$a_v$ and $a_{v'}$ approve each other, implying that they form a blocking pair for $N$ in this case, a contradiction.
	
	\medskip
	
	For \textsc{Individual Weak Stability}, the construction is the same and we set~$\ell = \alpha=2$. 
	For the forward direction of the proof of correctness note that as a blocking pair for $2$-individual weak stability needs to include two agents that are both unhappy in at least one layer and approve each other, a potentially blocking pair again needs to include two vertex agents corresponding to vertices unmatched by $M$. Then, the reasoning from above applies. 
	
	For the backward direction of the proof of correctness, recall that each $\ell$-individual weakly stable matching is also all-layers weakly stable and thereby that the reasoning from above still applies.  
\end{proof}
We leave it open which of our problems for strong and super stability that are NP-hard for asymmetric approvals are in FPT or in XP with respect to $\beta$.

\section{Conclusion}
We initiated the study of stable matchings with multilayer approval preferences.
We identified eleven stability notions and determined the computational complexity of deciding the existence of a stable matching for each notion.
While this task turned out to be NP-hard for just two or three layers for most of the notions (even if the analogous problem for strict preferences is polynomial-time solvable), we also identified several tractable cases, e.\,g., when ``similarity'' in the agents' preferences is assumed. 

For future work, note that we have posed several open questions throughout the paper, e.\,g., which of our problems become polynomial-time solvable if, within each layer, all agents approve the same agents. 
We also wonder for the two cases where we have polynomial-solvability for symmetric approvals but NP-hardness for asymmetric approvals, whether the problem is FPT with respect to the number of non-mutual approvals. 

On a more conceptual note, we have argued in the introduction that multilayer preferences also allow to model situations where fixed groups need to be matched to each other and each agent models a group. 
As groups can be of different sizes, it would be interesting to consider situations where each agent has a different number of preference relations. 
While pair and global stability seem no longer applicable, variants of individual stability still appear to be highly relevant.

It would also be interesting to consider multilayer variants of stable matching problems with ties and incomplete lists, which would notably generalize both the models studied by us and by Chen et al. \cite{DBLP:conf/sigecom/ChenNS18}. 
Thus, our strong intractability results already rule out the existence of efficient algorithm for many stability notions in this model.   

Regarding \cref{se:similarity}, one may also consider different similarity measures, e.\,g.\ (isomorphism-based) similarity of the approval graphs of the different layers.

Lastly, studying multilayer preferences in situations where agents shall be partitioned into groups of size larger than two (also known as hedonic games~\mbox{\cite{DBLP:reference/choice/AzizS16,DBLP:conf/aaai/Peters16}}) is a promising direction for future work.

\subsection*{Acknowledgments}
MB was supported by the DFG project MaMu (NI 369/19). NB was supported by the DFG project MaMu (NI 369/19) and by the DFG project ComSoc-MPMS (NI 369/22).  KH was supported by the DFG Research Training Group 2434 ``Facets of Complexity'' and by the DFG project FPTinP (NI 369/16). TK was supported by the DFG project DiPa (NI 369/21). 
This work
was started at the research retreat of the TU Berlin Algorithmics and Computational Complexity
research group held in Zinnowitz (Usedom) in September 2021.

\bibliographystyle{splncs04}

\end{document}